\documentclass[11pt]{article}

\usepackage{float}
\usepackage{amsmath}
\usepackage{fullpage}
\usepackage{amsfonts}
\usepackage{amssymb}
\usepackage{graphicx}
\usepackage{caption}
\usepackage{subcaption}
\usepackage{color}
\usepackage{graphics}
\usepackage{epstopdf}
\usepackage{dsfont}
\usepackage{times}
\usepackage{overpic}
\usepackage{framed}
\newtheorem{theorem}{Theorem}

\newtheorem{lemma}{Lemma}

\newtheorem{corollary}{Corollary}
\newtheorem{remark}{Remark}
\newtheorem{definition}{Definition}

\numberwithin{equation}{section}
\newenvironment{proof}[1][Proof]{\noindent\textbf{#1.} }{\ \rule{0.5em}{0.5em}}

\renewcommand{\epsilon}{\varepsilon}
\newcommand{\ket}[1]{\mathop{\left|#1\right>}\nolimits}

\newcommand{\bk}[2]{\langle #1 | #2 \rangle}

\newcommand{\abs}[1]{\left\vert #1 \right\vert}
\newcommand{\kb}[2]{| #1\rangle\!\langle #2 |}

\newcommand{\Tr}[2]{\mathop{{\mathrm{Tr}}_{#1}} (#2) }

\newcommand{\dg}{\dagger}

\newcommand{\mumax}{{\mu_{\text{max}}}}

\newcommand{\fl}[1]{{\lfloor #1 \rfloor}}
\newcommand{\cl}[1]{{\lceil #1 \rceil}}

\def\U{\mathcal{U}}

\def\bE{\mathbb{E}}
\def\Pr{\mathrm{Pr}}
\def\QIC{\mathrm{QIC}}

\def\bbE{\mathbb{E}}
\def\ZO{\{0,1\}}
\def\DISJ{\mathrm{DISJ}}
\def\AND{\mathrm{AND}}
\def\OR{\mathrm{OR}}

\makeatletter
\let\@copyrightspace\relax
\makeatother

\begin{document}

\title{Practical Quantum Appointment Scheduling}

\author{
Dave Touchette \thanks{Institute for Quantum Computing and Department of Combinatorics and Optimization,
University of Waterloo, and Perimeter Institute for Theoretical Physics.
} 
\and
Benjamin Lovitz \thanks{Institute for Quantum Computing and Department of Physics and Astronomy, University of Waterloo.} 
\and 
Norbert L\"{u}tkenhaus \thanks{Institute for Quantum Computing and Department of Physics and Astronomy, University of Waterloo, and Perimeter Institute for Theoretical Physics.}
}

\maketitle


\begin{abstract}
We propose a protocol based on coherent states and linear optics operations
for solving the appointment-scheduling problem.
Our main protocol leaks strictly less information about each party's input
than the optimal classical protocol, even when considering experimental
errors. Along with the ability to generate constant-amplitude coherent states
over two modes, this protocol requires the ability to transfer these modes
back-and-forth between the two parties multiple times with  low coupling
loss. The implementation requirements are thus still challenging.  Along the way, we develop new tools to study
quantum information cost of interactive protocols in the finite regime.
\end{abstract}

\section{Introduction}

In 2-party communication complexity, the main figure of merit is the minimum 
amount of communication required to perform a given distributed information-processing 
task. Say Alice is given some input $x$ and Bob is given some input $y$, then they wish to 
compute some relation $T$ evaluated at the joint input $(x,y)$, i.e.~they wish to 
both output an element of the set $T(x, y)$. Their goal is to minimize the amount of communication required to do so.  
If $x, y \in \{ 0 , 1 \}^n$ and $T(x, y) \subseteq \{0, 1\}^m$, then this complexity is at most $n+m$ bits: 
Alice can start by sending $n$ bits to Bob to communicate $x$, and then Bob can compute an element of $T(x, y)$
and transmit it back to Alice using $m$ bits.
Can they do significantly better?
In this work, we are interested in a variant of the communication complexity model, the information 
complexity model, which instead ask what is the minimum amount of \emph{information} Alice and 
Bob must \emph{leak} to each other about their inputs, irrespective of the amount of communication 
required to minimize this information leakage.

For both the communication and information complexity models, the complexity depends heavily on 
what resources are allowed and accounted for. Is communication done over classical or quantum channels? 
Are Alice and Bob allowed to flip random coins? Are they allowed to pre-share randomness or entanglement? 
In this work, we focus on how much advantage in terms of information leakage they can get by exchanging 
quantum rather than classical messages.
We study this in a quantum honest-but-curious type of model, in which we want the parties to exchange the correct messages, but they might collect as much information as possible about each other's input.
It is known that for some tailored problems, exponential savings 
are possible if Alice and Bob have access to perfect local quantum computers and perfect quantum 
communication channels~(see, for example, \cite{raz1999exponential, gavinsky2007exponential}). If we wish to limit Alice and Bob to quantum 
operations that should be experimentally accessible in the near future, can they still hope to 
achieve a quantum advantage in terms of information leakage?

We show that indeed they can. More precisely, we focus on quantum protocols requiring 
coherent state messages over two optical modes that are manipulated with linear optics 
operations and do not require any pre-shared entanglement or any quantum memory from 
honest participants. We compare such protocols with the best classical protocols for which we 
allow both local and shared randomness for free in order to minimize the information leakage. 
We also allow these classical resources to be used in our quantum protocols, appropriately accounting for them while quantifying information leakage. We find that indeed, 
with 
experimental parameters  that are challenging but should be reachable in the near future,
it is possible to obtain such a quantum advantage in terms of information leakage.
In fact, since we are mainly concerned with privacy here, Alice and Bob could be close 
to each other, in the same lab, and keep their inputs private but still have close-by set-ups 
which would perform much better than our data for clearly separated set-ups.

The problem we focus on is that of appointment scheduling: Alice and Bob each hold a 
calendar of their availabilities, and they wish to find a date of common availability, 
or agree that no such date exists. Viewing their inputs $x, y$ of available dates as subsets of a 
calendar $[n] = \{ 1, 2, \cdots, n \}$ on $n$ dates, they wish to output an element $i \in x \cap y$ if such an $i$ exists, or else 
output $\emptyset$ if $x \cap y = \emptyset$.
This problem, and in particular its binary variant, is one of the most well-studied problems in communication and information complexity.

It is known that quantum protocols can provide a quadratic 
speed-up in terms of information leakage for this problem~\cite{Buhrman:1998:QVC:276698.276713, HdW:2002, aaronson:2003}. 
It is also known that interaction is necessary to get an advantage over 
classical protocols~\cite{klauck2001interaction, JainRS03, braverman2015near}. As it turns out, for our protocols, 
interaction poses a challenge in a realistic experimental setting:
more interaction also implies more losses over the communication channels.
We show that there is nevertheless some regime 
for which we can obtain a quantum advantage.

Hence, our work is the first to propose an optical protocol that works with coherent states and maintains a quantum advantage in 
the more natural setting where Alice and Bob can directly interact.

\textbf{Related Works.} In Ref.~\cite{PhysRevA.89.062305}, Arrazola and L\"utkenhaus showed that a similar  practical quantum advantage 
was possible in terms of abstract cost of communication (the qubit size of the Hilbert space effectively used). 
The information complexity aspect of this protocol has been considered in Ref.~\cite{AT16}.
They studied a different communication model, the simultaneous message passing model, and a different problem, the equality function. 
In that model, Alice and Bob each send a simultaneous message to some referee who must then 
decide, using these messages only and no further information about Alice's and Bob's inputs, 
whether their inputs are equal. 
The advantage they show holds in the three party simultaneous message passing model (SMP) without 
shared randomness. However, considering the equality function in the direct interactive two-party model that
we focus on here, if we allow a single direct interaction between Alice and 
Bob (or even just a logarithmic length shared random string in the SMP model), they can solve the equality function at low cost.
A related SMP model coherent state protocol for evaluating the Euclidean distance between two real unit vectors was recently proposed in \cite{PhysRevA.95.032337}, and similar remarks apply.
Two other recently-proposed communication protocols which use coherent states are quantum retrieval games \cite{PhysRevA.93.062311} and quantum money schemes \cite{PhysRevA.95.062334}.

\textbf{Organization.} The remainder of the paper is structured as follows. In the next section we describe our practical 
quantum protocol for appointment scheduling, and analyse its 
behavior in an idealized setting. In the following section, we analyse it in a more 
realistic experimental setting, accounting for errors,  and explore the parameter 
space to find a reasonable regime in which our quantum protocol performs better 
than any classical protocol. We conclude by discussing our findings and opportunities 
for future work.

In  Appendix~\ref{leakage} we formally define the information leakage and introduce some properties which we use in Appendices~\ref{leakage_ideal} and~\ref{leakage_exp} to bound the information leakage of our protocol. In Appendix~\ref{mappingrev} we review a mapping proposed in~\cite{PhysRevA.95.032337} from pure state communication protocols to coherent state protocols, which we use in Appendices~\ref{app2} and~\ref{cohgrov} to develop two more coherent state appointment scheduling protocols.

\section{Coherent-state Protocol}\label{coh_prot}

In the idealized setting of quantum communication complexity, a protocol that achieves 
the quadratic quantum advantage, up to logarithmic terms, for appointment scheduling, 
is that of~\cite{Buhrman:1998:QVC:276698.276713}, essentially performing a distributed version of Grover 
search~\cite{grover1996fast, boyer1996tight}. Alice performs the ``inversion about the mean'' Grover 
iterations to find an intersecting date of availability, and she collaborates with Bob in order 
to implement the Grover ``oracle calls''.

For an $n$-date calendar, obtaining the full quadratic quantum advantage requires $\tilde{\Theta} (\sqrt{n})$ 
rounds of communication, while an improvement to $\tilde{\Theta} (\frac{n}{r})$ communication 
and information leakage requires $r$-round protocols~\cite{JainRS03, braverman2015near}, for $r \leq \sqrt{n}$.

In Ref.~\cite{PhysRevA.95.032337} a general mapping is proposed from any pure state quantum protocol to an analogous coherent state protocol (reviewed in Appendix~\ref{mappingrev}). In Appendix~\ref{cohgrov} we implement this mapping for the distributed Grover's search protocol to obtain essentially 
a quadratic quantum advantage in terms of  information leakage. 
Our implementation finds an efficient way to perform the “distributed oracle calls” for such a protocol.
Note that experimental implementations of Grover search using optics have already been performed, e.g., in Ref.~\cite{bhattacharya2002implementation}.
However, such an $\Theta (\sqrt{n})$-round protocol requires Alice to interfere $n$ modes 
together for each Grover iteration, and thus the experimental complexity of such a protocol 
grows very quickly with $n$.
Also taking into account experimental errors, such an 
approach quickly becomes impractical.

We instead focus on an alternative approach in which interfering two optical modes is 
always sufficient. The actual quantum part of the protocol focuses on a single date and wishes to determine 
whether this is an intersecting date. Viewing Alice's and Bob's input for the quantum part 
as single bits, $a$ and $b$, respectively, they thus wish to compute $AND (a, b)$. Let us call this 
quantum subroutine $\widetilde{\Pi}_A$, which can either output ``$0$'', ``$1$'' or ``Inconclusive''.

The approach we have taken to {ensure} that our protocol has low information leakage is similar in spirit to the one taken in Ref.~\cite{braverman2015near}: first subsample many dates to {ensure} that there are not too many intersections, and then run {a date-wise $AND$ protocol} that is only guaranteed to have low information leakage when the probability to find an intersection is low. Note that Jain, Radhakrishnan and Sen had proposed such a low information protocol for $AND$, which we review in Appendix~\ref{app2}. We also consider the coherent state mapping applied to this protocol in Appendix~\ref{app2}, but find that our protocol $\widetilde{\Pi}_A$ has lower information leakage, and in addition, it appears to be simpler from an experimental point-of-view. It is also easily extendable to a multi-party setting, by having the other parties act similarly to Bob.

To describe the protocol $\widetilde{\Pi}_A$, recall the mathematical definition of a two 
mode coherent state $\ket{\alpha, \beta} = \ket{\alpha} \otimes \ket{\beta}$, with
\begin{align*}
	\ket{\alpha} = \exp(-|\alpha|^2/2) \sum_{k=1}^\infty \frac{\alpha^k}{\sqrt{k !}} \ket{k},
\end{align*}
for $\alpha \in \mathbb{C}$ (we only make use of $\alpha \in \mathbb{R}$), as well as 
the action of the beam-splitter $R_\theta$ at angle $\theta$ on such a state:
\begin{align*}
	R_\theta \ket{\alpha, \beta} = \ket{\cos (\theta) \alpha - \sin (\theta) \beta , \sin (\theta) \alpha  + \cos (\theta) \beta  }.
\end{align*}
Depending on the parameter $r$, corresponding to the number of rounds of interaction, the angle
of the beam-splitter in the protocol is $\theta_r = \frac{\pi}{2r}$, so that $r \theta_r = \frac{\pi}{2}$.

\begin{framed}

\textbf{Protocol $\widetilde{\Pi}_A$} on inputs $a , b \in \{0, 1 \}$:

In the initialization phase, Alice prepares a two-mode register $C$ in state $\ket{\alpha, 0}$.
Then,  for rounds $i=1$ to $r$:

\begin{enumerate}

\item On $a=0$, Alice applies the identity map to register $C$ 
and sends the transformed state to Bob. On $a=1$, Alice instead  passes the two modes of register $C$ through the beamsplitter $R_{\theta_r}$ and then sends register $C$ to Bob.

\item On $b=0$,  Bob discards the state of register $C$, replaces it with a fresh 
copy of $\ket{\alpha, 0}$ and sends it to Alice. On $b=1$, Bob applies the identity map to register $C$ and returns it to Alice.
\end{enumerate}

After $r$ rounds, Alice measures each mode of register $C$ with single photon threshold detectors and 
communicates the result to Bob.
They generate their output as follows:
\begin{itemize}
\item If only the first mode clicks, they output ``$0$''.
\item If only the second mode clicks, they output  ``$1$'' .
\item If neither mode clicks, they output ``Inconclusive''.
\end{itemize}

\end{framed}

With no  losses, the amplitude $\alpha$ of the re-injected states can stay the same throughout.
Re-injecting states with decreasing amplitudes $\alpha_i$ is however useful when coherent states are transmitted back-and-forth over lossy channels, as studied in the next section.

In the ideal setting, this protocol evolves as follows on the different inputs.

\begin{samepage}
\textbf{Evolution of $\widetilde{\Pi}_A$ for different inputs:}

\begin{flushleft}
\textbf{On (0, 0):}
\begin{tabular}{l l l l}
$\ket{\alpha, 0} $ & $\rightarrow_A \ket{\alpha, 0}$ &$\rightarrow_B \ket{\alpha, 0} $ & $\rightarrow_A \cdots $
\end{tabular}\\
\textbf{On (0,1):}
\begin{tabular}{l l l l}
$\ket{\alpha, 0} $ & $\rightarrow_A \ket{\alpha, 0} $ & $\rightarrow_B \ket{\alpha, 0} $ & $\rightarrow_A \cdots $
\end{tabular}\\
\textbf{On (1, 0): }
\begin{tabular}{l l}
$\ket{\alpha, 0} \rightarrow_A \ket{\cos (\theta) \alpha, \sin (\theta ) \alpha}  \rightarrow_B \ket{\alpha, 0} \rightarrow_A \ket{\cos (\theta) \alpha, \sin (\theta ) \alpha} \rightarrow_B \ket{\alpha, 0} \rightarrow \cdots$  \\
\hspace{.2in}		$\cdots \rightarrow \ket{\alpha, 0}$ &
\end{tabular}\\
\textbf{On (1, 1):}\\
\begin{tabular}{l l}
$\ket{\alpha, 0} \rightarrow_A \ket{\cos (\theta) \alpha, \sin (\theta ) \alpha} \rightarrow_B \ket{\cos (\theta) \alpha, \sin (\theta ) \alpha} $ \\
\hspace{.2in}		 $\rightarrow_A \ket{\cos (2 \theta) \alpha, \sin (2 \theta ) \alpha} \rightarrow_B \ket{\cos (2 \theta) \alpha, \sin (2 \theta ) \alpha} $ \\
 \hspace{.3in} \vdots  \\ 
\hspace{.2in}$\rightarrow_B \ket{\cos (\frac{\pi}{2}) \alpha , \sin (\frac{\pi}{2}) \alpha} = \ket{0, \alpha} $.
\end{tabular}
\end{flushleft}
\end{samepage}

On (0,0) and (0,1) Alice and Bob's manipulations leave the state unchanged. On (1,0) Alice rotates the state and then Bob replaces it with $\ket{\alpha,0}$ in each round. On (1,1) Alice and Bob's manipulations bring the state to $\ket{0,\alpha}$ after $r$ rounds.

Using $\widetilde{\Pi}_A$, we recursively define a conclusive protocol $\Pi_A$ for $AND$ that only 
outputs ``$1$'' after a classical verification that the date indeed intersects.

\begin{framed}

\textbf{Protocol} $\Pi_A$ on inputs $ a , b \in \{0, 1 \}$:

\begin{enumerate}
\item Run Protocol $\widetilde{\Pi}_A$.
\item If $\widetilde{\Pi}_A$ returns ``0'', return output ``0''.
\item If $\widetilde{\Pi}_A$ returns ``1'', Alice and Bob exchange $a$ and $b$ and return $AND(a, b)$ as output.
\item If $\widetilde{\Pi}_A$ returns ``Inconclusive,'' restart $\Pi_A$.
\end{enumerate}

\end{framed}

Finally, we describe a protocol $\Pi_D$ for appointment scheduling on $n$-dates that works by running as 
a subroutine protocol $\Pi_A$ for determining if a single date intersects. It either outputs a 
date of intersection ``$i$'' $ \in [n]$ or, if they believe no such date exists, ``$\emptyset$''.
We abuse notation and write Alice's input set $x$ as an n-bit indicator variable, with $x_i=1$ if and only if $i \in x$, and similarly for Bob's $y$.

\begin{framed}

\textbf{Protocol} {$\Pi_D$} on inputs $x, y \in \{0, 1 \}^n$:

\begin{enumerate}
\item Using shared randomness, publicly sample $s$ dates with replacement. Denote this 
date set by $S$.
\item Alice sends $x_i$ to Bob for each $i \in S$.
\item If Bob find any  $i \in S$ with $x_i = y_i = 1$, he sends the smallest such $i$ to Alice, and both 
output this $i$. Else, they continue.
\item Run date-wise  the $\Pi_A$ protocol for all dates outside of $S$. 
\item If they find any $i$ such that $AND(x_i, y_i) = 1$, both output the smallest such~$i$. 
\item If they do not find any such $i$, output ``$\emptyset$''.
\end{enumerate}

\end{framed}

This protocol clearly solves the appointment scheduling problem. 
While it does not guarantee to find the earliest intersecting date, its output is nevertheless biased towards such an early date.
Note that this protocol is parameterized by $s \in \mathbb{N}$, the size of a sample from the input that Alice and Bob exchange.
As the number of rounds $r$ increases, Alice is required to interfere 
these modes through a beam-splitter at decreasing angle $\theta_r = \frac{\pi}{2r}$. 
In the ideal scenario, this allows the information leakage per date to decrease to $\Theta (\frac{\log r}{r})$ if there is no intersection.
Although the number of signals exchanged is $\Omega (nr)$, we can still prove a 
guaranteed quantum advantage in terms of information leakage.
The quantum subroutine $\widetilde{\Pi}_A$ is independent of the calendar size $n$, so only the global routine $\Pi_D$ depends on the calendar size.
We are thus able to handle huge inputs and still 
obtain a quantum advantage.
In particular, this allows us 
to avoid some finite size effects of working with small values of $n$.

This protocol solves the appointment-scheduling problem with no 
errors whenever the overall optical set-up is ideal, in particular if the single photon threshold detectors are perfect.
We also prove the following about the information leakage  of this protocol when run over 
lossless channels with perfect detectors.
See Appendix~\ref{leakage} for precise definitions of the information leakage ($\mathrm{QIC}$, as defined in Ref.~\cite{touchette2015quantum, LauriereT:2016}).
Here and throughout, $h(\cdot)$ is the binary entropy function,  $\log$ is taken in base $2$  and the natural logarithm is denoted $\ln$.

\begin{theorem}\label{QIC_general_ideal}
The following holds for the protocol $\Pi_D$ when it is run in an ideal experimental set-up. 
The protocol never errs and the information leakage satisfies
\begin{align*}
\mathrm{QIC} (\Pi_D) \leq & \;s + \log s + 1 \\
	&+ \frac{n}{ 1 - \exp (- \abs{\alpha}^2)} \; \max \Bigg[ \frac{2(2r+1)}{n} \; , \\
&\hspace{0.5in}  h\left(\frac{1}{2} (1 - F(r, \alpha))\right)
 + 2(2r+1)  \; h\left(\frac{2  \ln n }{s } + \frac{1}{n}\right)\Bigg]
\end{align*}
for all values of $n$ and $s$ satisfying $n\geq 4$ and $8\ln(n)\leq s \leq n$,
and in which
\begin{align*}
F(r, \alpha) = \exp \left[- r \abs{\alpha}^2 \left[1 - \cos \left(\frac{\pi}{2r}\right)\right] \right].
\end{align*}

\end{theorem}


This upper bound on the quantum information leakage of our protocol in the ideal case can be optimized by minimizing the expression over $s$, $r$, and $\alpha$.
We are interested in the regime where $\alpha \approx 1$, $s \approx n^{2/3}$ and $r \approx n^{1/3}$; in that regime, all terms are at most $n^{2/3}$ up to logarithmic factors. (Notice that the $h(1/2(1-F(r, \alpha)))$ term inside the bracket scales like $1/r$ up to a logarithmic factor. Hence, the choice for $r$ is motivated by the desire to keep $n \; h(1/2(1-F(r, \alpha)))$ scaling as $n^{2/3}$, while a similar motivation for the next term motivate the choice for $s$.) We get the upper bound  $\QIC(\Pi_D) \in \widetilde{\mathcal{O}}(n^{2/3})$.  In contrast, the information leakage of any classical protocol is known from Ref.~\cite{braverman2013information}, and it is at least 0.48n for all n.  We plot this asymptotic improvement in Figure~\ref{fig:ideal_qu_adv}. 

We do not get the full quadratic speed-up that qubit-based quantum protocols can achieve. This restriction comes from the fact that we perform classical subsampling of size $s$. If we were guaranteed that there was at most a logarithmic number of intersections, then we could avoid subsampling (i.e., pick $s = 0$ in our protocol) and obtain the full quadratic speed-up, up to logarithmic terms, by choosing $r \approx n^{1/2}$. This follows since the argument of $h$ in the last term in the bracket would then be guaranteed by assumption (rather than by the subsampling, as done currently) to be essentially $1/n$ up to logarithmic factors. Another approach, similar to what was done in Ref.~\cite{braverman2015near}, would be to take the subsampling set to be quasi-linear in size, up to logarithmic terms, and then look for an intersection using a quantum protocol with low information leakage, say $\sqrt{n}$. But then, for the problem we wish to solve, one might just as well use this low information protocol to solve the appointment scheduling problem itself. We propose such a protocol in Appendix~\ref{cohgrov}, achieving information leakage at most $\sqrt{n}$ up to logarithmic factors. However, that protocol requires interfering $n$ optical modes together in a beamsplitter.
Such an approach quickly becomes impractical,
hence our choice of protocol with the classical subsampling.

\section{Accounting for Experimental Errors}
\label{sec:account_exp_error}

Our developments so far have assumed an ideal experimental setup. However, if we wish to 
model any practical implementation, we must take into account experimental errors. We now 
show how to slightly modify our protocol so that it is robust against these errors. 
We focus on three main sources of errors: losses incurred from limited channel transmissivity and fiber coupling, the efficiency of the threshold detectors, and the dark count probability of the threshold detectors.

The effect of the loss, characterized by the transmissivity parameter $\eta$, is the 
transformation $\ket{\alpha, \beta} \rightarrow \ket{\sqrt{\eta} \alpha, \sqrt{\eta} \beta}$ 
for one message in the protocol, and takes into account both transmission loss (which depends on the physical distance between Alice and Bob), and coupling loss (which is independent of the distance between Alice and Bob). This can be 
compensated in $\widetilde{\Pi}_A$ by taking 
larger values of $\alpha$ for the initial state as well as the reinjected state,
at the cost of an increase in information leakage.

The effect of detector efficiency, characterized by the efficiency parameter $\eta_{det}$, is similar to a loss when inputting the signals into the detectors, and is the transformation $\ket{\alpha, \beta} \rightarrow \ket{\sqrt{\eta_{det}} \alpha, \sqrt{\eta_{det}} \beta}$.

The effect of dark counts is to make the ideal protocol for $AND$ prone to errors. 
Indeed, the ideal protocol might be inconclusive, if no photon is detected, but 
whenever it returns output ``$0$'' or ``$1$'', this output can be trusted. It is no 
longer the case if there are dark counts: it is possible to get output ``$0$'' if $AND(a, b) = 1$, 
and vice-versa. 

Here is our modified protocol $\widetilde{\Pi}_A^\prime$; $\Pi_A$ and $\Pi_D$ stay unchanged, apart from now running $\widetilde{\Pi}_A^\prime$ rather than $\widetilde{\Pi}_A$ as a subroutine.

\begin{framed}

\textbf{Protocol $\widetilde{\Pi}_A^\prime$} on inputs $a , b \in \{0, 1 \}$:
Given $\alpha_{out}$ and $\eta$, let $\alpha_0=\frac{\alpha_{out}}{\eta^r}$ and $\alpha_1=\frac{\alpha_{out}}{\eta^{r-1/2}},\alpha_2=\frac{\alpha_{out}}{{\eta^{r-3/2}}}, \dots, \alpha_{r-1}=\frac{\alpha_{out}}{{\eta^{3/2}}}, \alpha_r=\frac{\alpha_{out}}{{\eta^{1/2}}}$.

In the initialization phase, Alice prepares a two-mode register $C$ in state $\ket{\alpha_0, 0}$.
Then, for rounds $i=1$ to $r$:

\begin{enumerate}

\item On $a=0$, Alice applies the identity map to register $C$ 
and sends the transformed state to Bob. On $a=1$, Alice instead passes the two-mode of register $C$ through the beamsplitter $R_{\theta_r}$ and then sends register $C$ to Bob.

\item On $b=0$, for round $i$, Bob discards the state of register $C$, replaces it with a fresh 
copy of $\ket{\alpha_i, 0}$ and sends it to Alice. On $b=1$, Bob applies the identity map to register $C$ and returns it to Alice.
\end{enumerate}

After $r$ rounds, Alice measures each mode of register $C$ with single photon threshold detectors and 
communicates the result to Bob.
They generate their output as follows:
\begin{itemize}
\item If only the first mode clicks, they output ``$0$''.
\item If only the second mode clicks, they output  ``$1$'' .
\item If neither mode clicks or both modes click, they output ``Inconclusive''.
\end{itemize}
\end{framed}

The extra classical verification in $\Pi_A$ in case of a ``$1$'' output of $\widetilde{\Pi}_A^\prime$ 
is to ensure that at $\Pi_A$ and then at $\Pi_D$ level, errors can only be ``one-sided'': it is possible that an 
intersecting date is not detected as such, but a non-intersecting date is never thought to be intersecting. 
This property of the error at the date level severely limits propagation of errors. We show 
in Appendix~\ref{leakage_exp} that, for dark count probability $p_{dark}$ for the measurement at the end of each execution of $\widetilde{\Pi}_A^\prime$, the overall error probability 
of $\Pi_D$ is $p_{dark}$ (rather than a bound $ \approx n \cdot p_{dark}$ obtained using the union bound when this extra check is not performed).

In particular, and in contrast to the practical fingerprinting protocol of Ref.~\cite{PhysRevA.89.062305}, the 
input size for which we can achieve a quantum advantage is not limited by $p_{dark}$ for us. It is rather the loss parameter $\eta$ which has 
a much bigger impact here, since quantum advantage for appointment scheduling requires interaction, and for $r$ 
rounds of interaction, the global effect of the loss is essentially $\eta^{2r}$. The following theorem provides 
bounds on the information leakage when the protocol is run while taking such experimental errors into account.

\begin{theorem}\label{QIC_general_exp}
The following holds for the protocol $\Pi_D$ when run with loss parameter $\eta$, dark count probability $p_{dark}$, and detector efficiency $\eta_{det}$.
With  $\alpha_0=\frac{\alpha_{out}}{{\eta^r}}$ and $\alpha_1=\frac{\alpha_{out}}{{\eta^{r-1/2}}},\alpha_2=\frac{\alpha_{out}}{{\eta^{r-3/2}}}, \dots, \alpha_{r-1}=\frac{\alpha_{out}}{{\eta^{3/2}}}, \alpha_r=\frac{\alpha_{out}}{{\eta^{1/2}}}$ in $\widetilde{\Pi}_A^\prime$,
the protocol $\Pi_D$ never outputs a date which is not intersecting, and the probability that the output is $\emptyset$ when there is an intersecting date is at most $p_{dark}$. The information leakage satisfies
\begin{align*}
\QIC (\Pi_D) \leq & \; s + \log s + 1 + \frac{2n}{ 1 - p} \; p_{dark} \\
		&+ \frac{n}{ 1 - p} \max \Bigg[\frac{2(2r+3)}{n}, \\
&\quad \quad \quad h\left(\frac{1}{2} (1 - \tilde{F} (r, \alpha_{out}, \eta))\right)+ 2(2r+3)  \; h\left(\frac{2  \ln n }{s } + \frac{1}{n}\right)\Bigg]\\
\end{align*}
for all values of $n$ and $s$ satisfying $n\geq 4$ and $8\ln(n)\leq s \leq n$,
and in which
\begin{align*}
\tilde{F} (r, \alpha_{out}, \eta) &= \exp \left[\frac{-(\eta^{-2r} - 1)}{(1- \eta^2)} \abs{\alpha_{out}}^2 \left[1 - \cos \left(\frac{\pi}{2r}\right)\right] \right],\\
p&=e^{- \eta_{det} \abs{\alpha_{out}}^2}(1-p_{dark})^2 +(1-e^{-\eta_{det} \abs{\alpha_{out}}^2}+e^{- \eta_{det} \abs{\alpha_{out}}^2}p_{dark})p_{dark}.
\end{align*}
($p$ is the probability of an inconclusive outcome.)
\end{theorem}


\begin{figure}[!t]
\centering
\includegraphics[width=400pt]{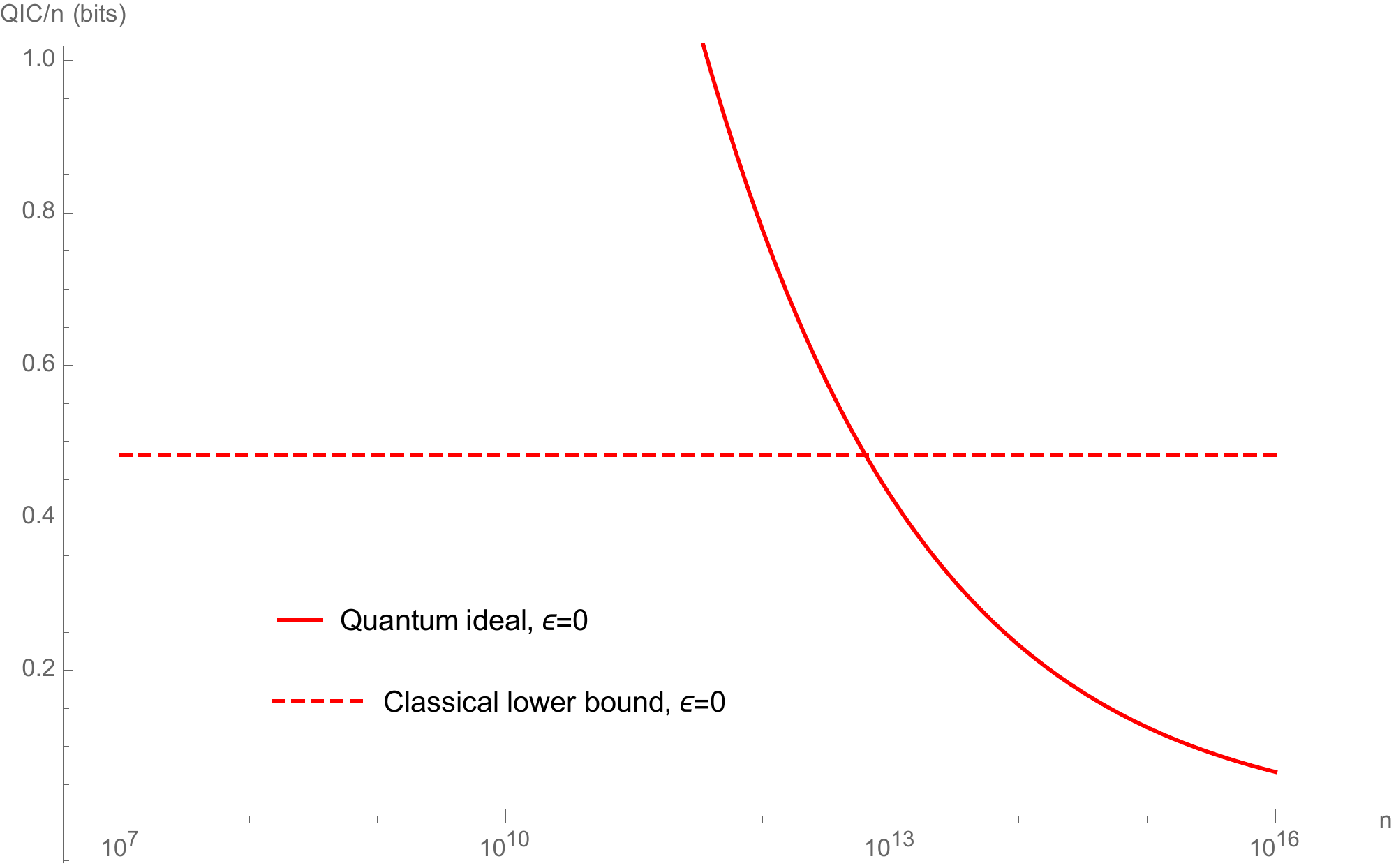}
\caption{This figure shows the $\mathcal{O}(n^{2/3})$ limiting behaviour of our quantum protocol in comparison with the $\Omega(n)$ classical lower bound in the ideal setting for zero-error. We have chosen the number of rounds $r = n^{1/3}$, the coherent state amplitude $\alpha =  1$ and subsample size $s = n^{2/3}$. The information leakage ($\QIC$) measured in bits divided by the input size $n$ is plotted on the y-axis, and the input size $n$ on the $x$-axis.}
\label{fig:ideal_qu_adv}
\end{figure}

\begin{figure}[!t]
\centering
\includegraphics[width=400pt]{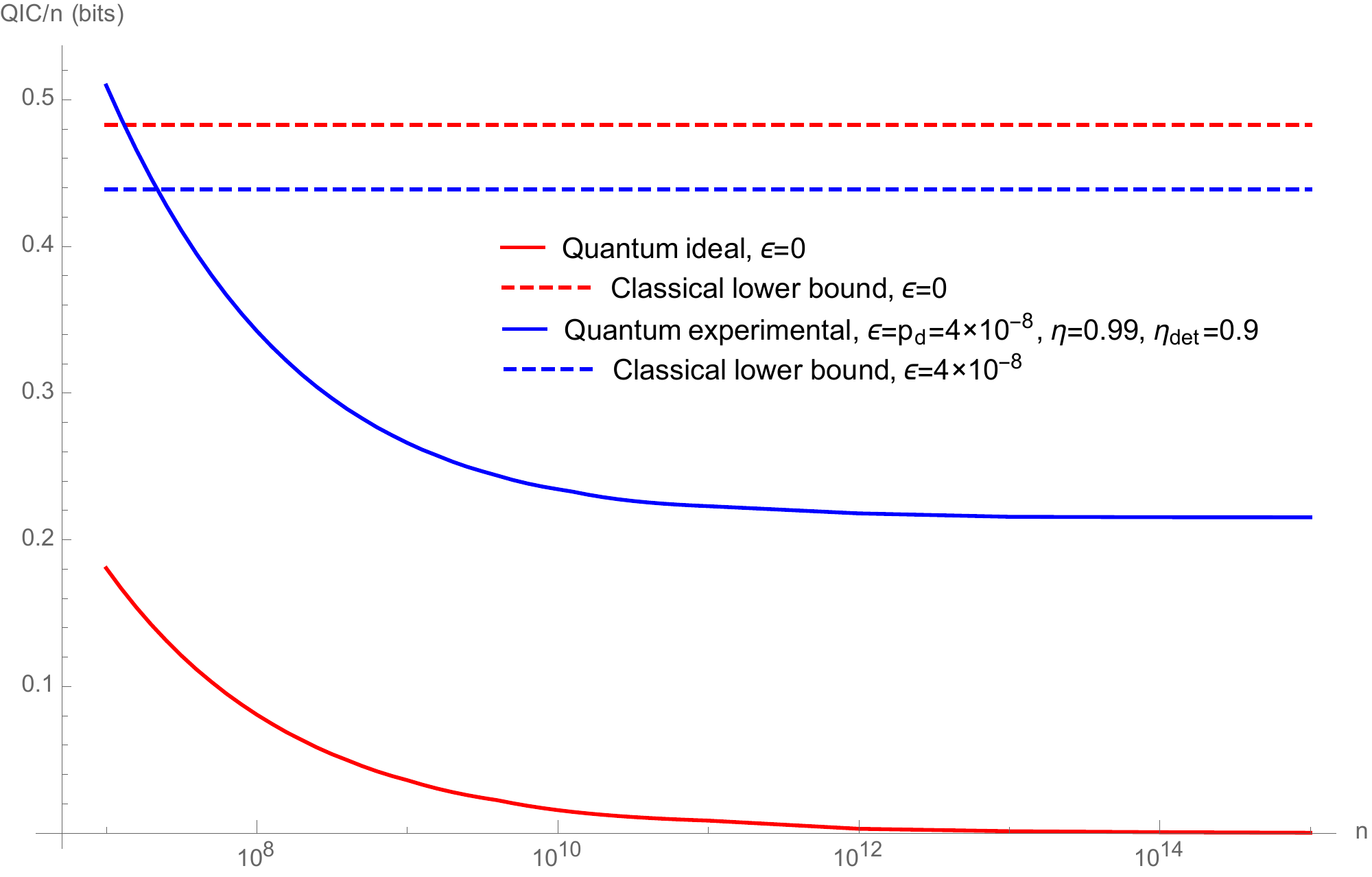}
\caption{This figure depicts the quantum advantage in terms of information leakage ($\QIC$) measured in bits, and compares both the classical lower bounds for zero-error and $\epsilon$-error protocols, to both the quantum upper bound for an ideal experimental set-up and for set-ups accounting for experimental errors with the following parameters: transmissivity $\eta = 0.99$, dark count probability $p_{dark} = \epsilon =4 \times 10^{-8}$, and detection efficiency $\eta_{det}=0.9$.
At each point we have optimized over $s$, $\alpha$, and $r$. The value of $s$ obtained by our optimization decreases from around $s=0.1 n$ to $s=0.001n$ as $n$ increases from $10^{7}$ to $10^{11}$. The optimized value of $r$ increases with $n$ from around $r=30$ to around $r=100$, and the optimized value of $\alpha$ remains near $\alpha=1$. The information leakage divided by the input size $n$ is plotted on the y-axis, and the input size $n$ on the $x$-axis. Note that the classical information leakage lower bound with non-zero error could probably be made much closer to the one at zero error by a careful analysis of Refs~\cite{dagan2016trading, braverman2013information} in the finite regime.}
\label{fig:pract_qu_adv}
\end{figure}

For appropriately chosen $s$ and small $p_{dark}$, the term $h\left(\frac{1}{2} (1 - \tilde{F}(r, \alpha_{out}, \eta) )\right)$ is the limiting one. If $\eta<1$, this term as a function of $r$ is limited by a trade-off between $\approx \eta^{-2r}$ and $\approx 1/r^2$.
For constant $\eta$, there is an optimal value of $r$ for this term, independent of $n$ and $\alpha$. Hence, for constant $\eta$ and $\alpha$, this term, even optimized over $r$, does not decrease as $n$ increases, and we can at best expect quantum advantage by a constant multiplicative factor in terms of information leakage. Taking $s \approx n^{2/3}$, this term as well as the term $2n \; p_{dark} / (1-p)$ are the only two growing linearly with $n$, hence they are the limiting ones for large $n$.

In Figure~\ref{fig:opt_r_and_alpha} we plot $r$ and $\alpha$ as a function of $\eta$, optimized to minimize the information leakage given in Theorem~\ref{QIC_general_exp}, under the parameters $n=10^{15}$, $p_{dark}=4\times 10^{-8}$, and $\eta_{det}=0.9$. Note that the tail-ends of these plots do not continue the behaviour established at lower values of $\eta$, which can be intuited from the fact that as $\eta \rightarrow 1$ the optimal $r$ for the term $h\left(\frac{1}{2} (1 - \tilde{F}(r, \alpha_{out}, \eta) )\right)$ grows large, so the other terms in Theorem~\ref{QIC_general_exp} become significant and must now be taken into account.

If $\eta \approx 1 - \frac{1}{r}$ but $p_{dark} > 0$, then we could improve on the quantum advantage to $\widetilde{\Theta}(\frac{n}{r})$ (for $r \leq n^{1/3})$ while keeping the same quantum subroutine but modifying the global classical processing.
See Appendix~\ref{sec:modif-one-over-r} for details.

{
In Figure~\ref{fig:pract_qu_adv} we plot the information leakage of our quantum protocol under the experimental imperfections {$\eta =0.99$}, $p_{dark} = 4 \times 10^{-8}$, and $\eta_{det}=0.9$ versus the classical lower bound. In contrast to the plot of Figure~\ref{fig:ideal_qu_adv}, the optimized number of rounds remains small so that the total loss does not increase too much with $n$ and completely degrade the system. For this reason, asymptotically, $r$ must scale constant in $n$. Unfortunately, this implies that for any fixed $\eta<1$ our protocol has asymptotic information leakage $\Theta(n)$ (just like the classical lower bound). However, we see in Figure~\ref{fig:pract_qu_adv} that it still gives rise to a quantum advantage by a factor of 2.} We also plot an optimized quantum advantage obtained by our protocol as a function of $\eta$ in Figure~\ref{fig:adv_vs_eta}. In particular, we note that by optimizing $r$ and $\alpha$ and picking appropriate subsampling size $s$, we can get a quantum advantage starting around $\eta \approx 0.975$ for $\eta_{det}=0.9$ and $p_{dark}=4\times 10^{-8}$. Also, for $\eta = 0.999$, $\eta_{det}=0.9$, and $p_{dark}=4\times 10^{-8}$ we get an improvement by a factor of more than 14!
\begin{figure}
\centering
\begin{subfigure}{.5\textwidth}
  \centering
  \includegraphics[width=1.0\linewidth]{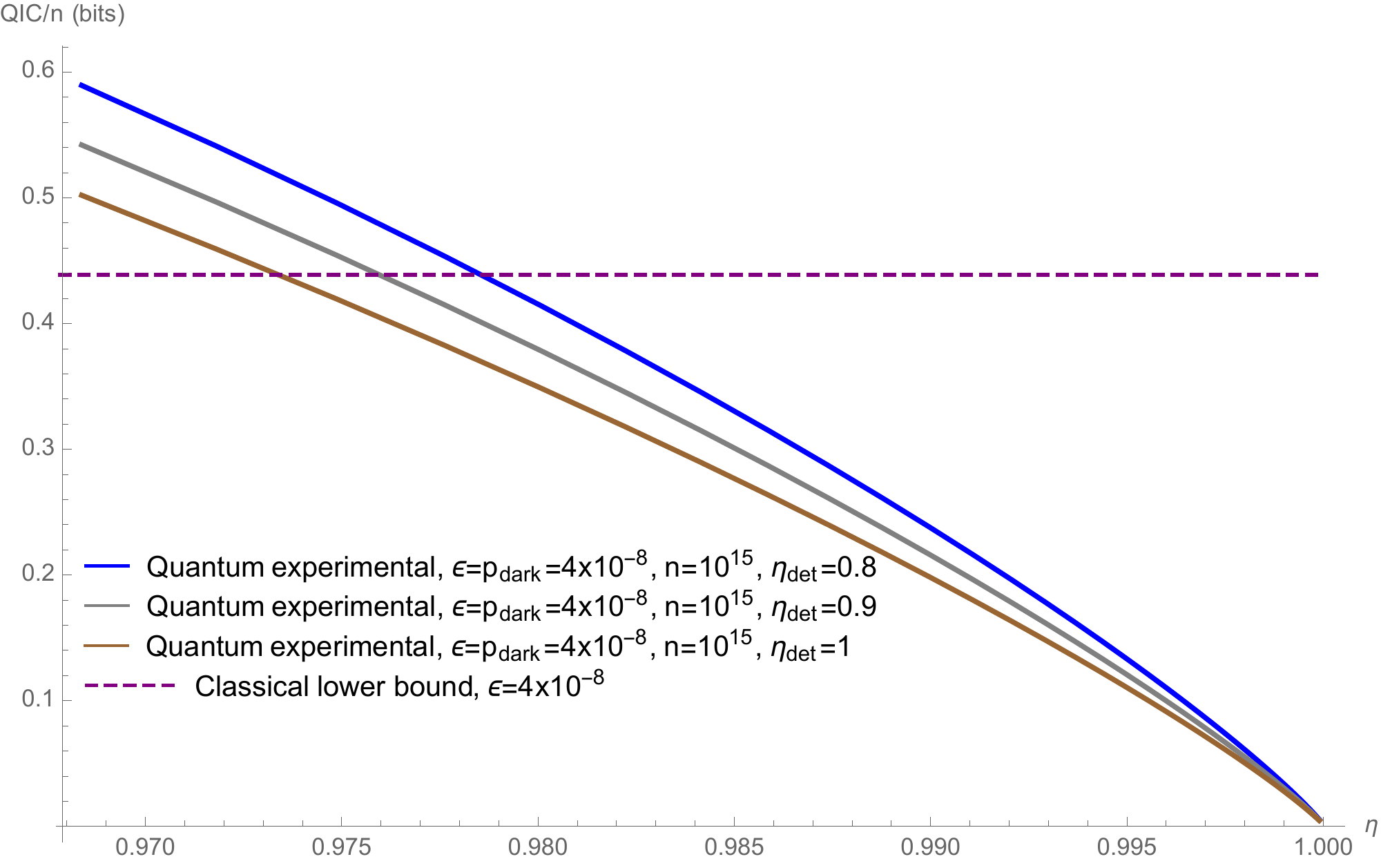}
\end{subfigure}%
\begin{subfigure}{.5\textwidth}
  \centering
  \includegraphics[width=1.0\linewidth]{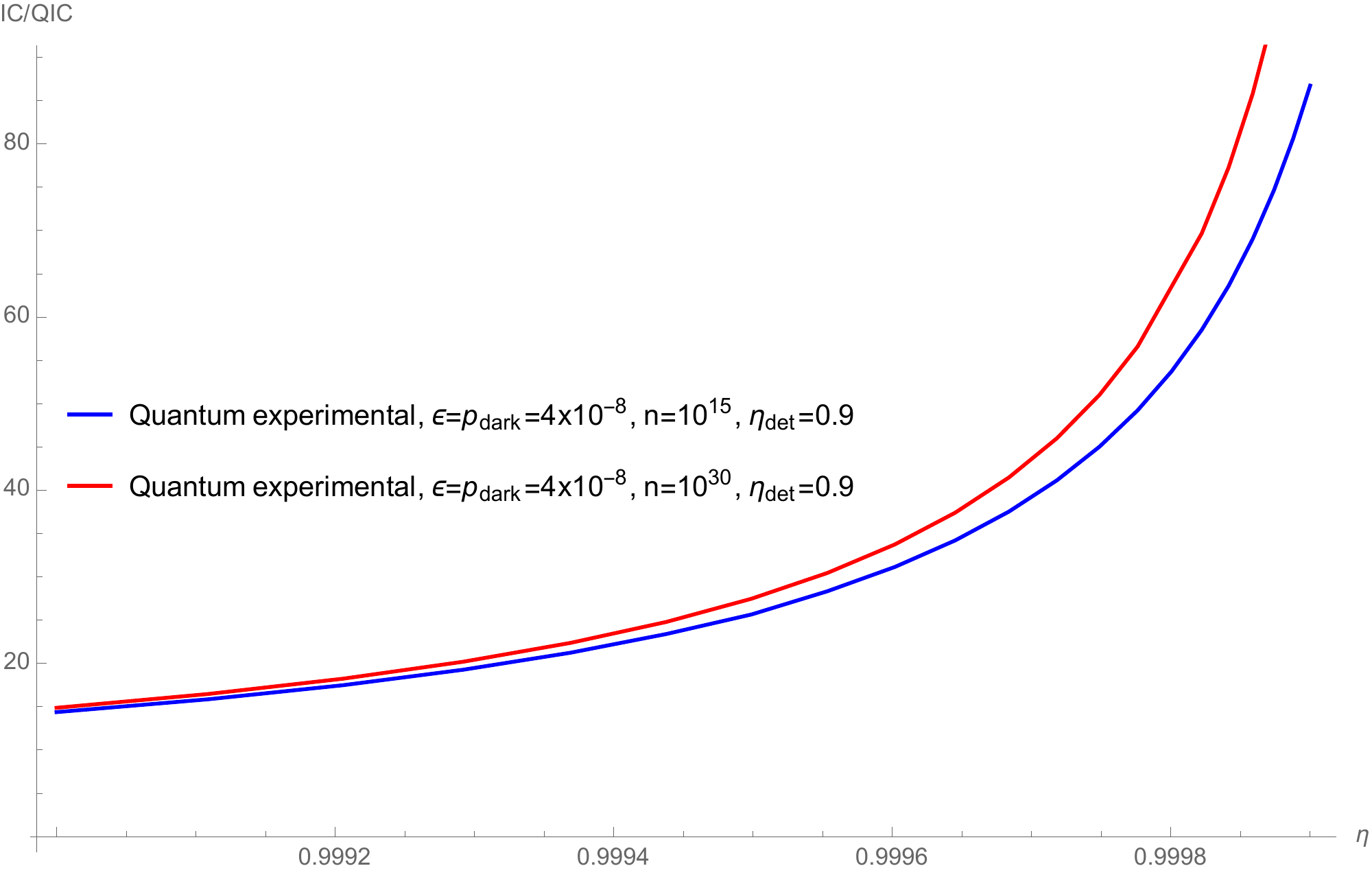}
\end{subfigure}
\caption{These figures depict the quantum advantage in terms of information leakage ($\QIC$) measured in bits as a function of the loss parameter $\eta$. On the left, it makes a similar comparison to the classical protocols as in Figure~\ref{fig:pract_qu_adv}, but for different values of $\eta_{det}$, and on the right it plots the classical over quantum ratio for $n=10^{15}$ and $n=10^{30}$. We have used $s=0.001 n$, which we have found performs very well for this range of $\eta$. We can see that for $p_{dark}=4\times 10^{-8}$ and $\eta_{det}=0.9$, we begin to get a quantum advantage for $\eta \approx 0.975$, and the advantage grows as $\eta$ goes to one. In particular, for $\eta = 0.999$, we get an improvement by a factor of more than 14! Note that we have implicitly optimized $r$ and $\alpha$ as functions of $\eta$ here to obtain the greatest advantage.}
\label{fig:adv_vs_eta}
\end{figure}
\begin{figure}
\centering
\begin{subfigure}{.5\textwidth}
  \centering
  \includegraphics[width=1.0\linewidth]{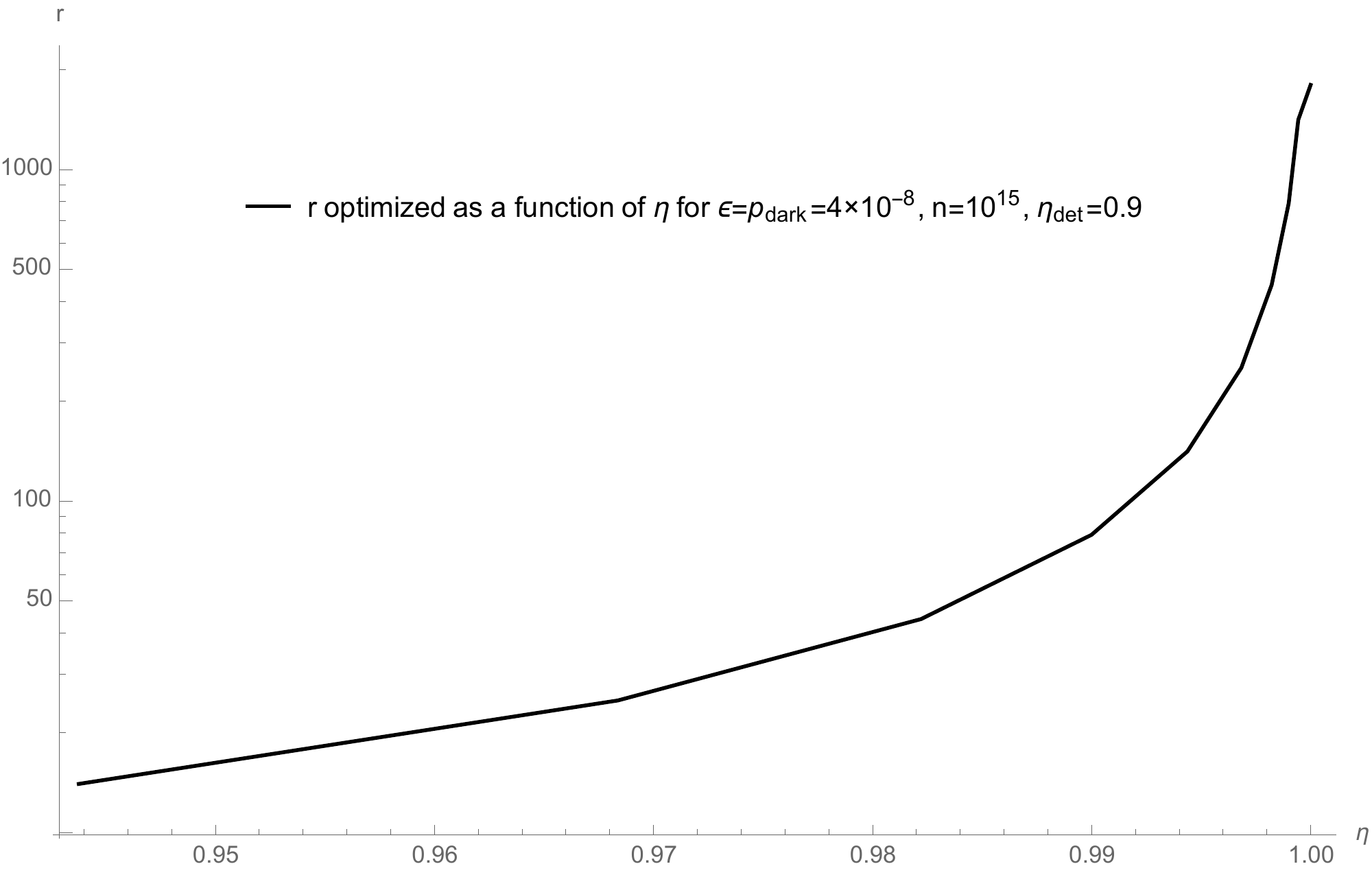}
\end{subfigure}%
\begin{subfigure}{.5\textwidth}
  \centering
  \includegraphics[width=1.0\linewidth]{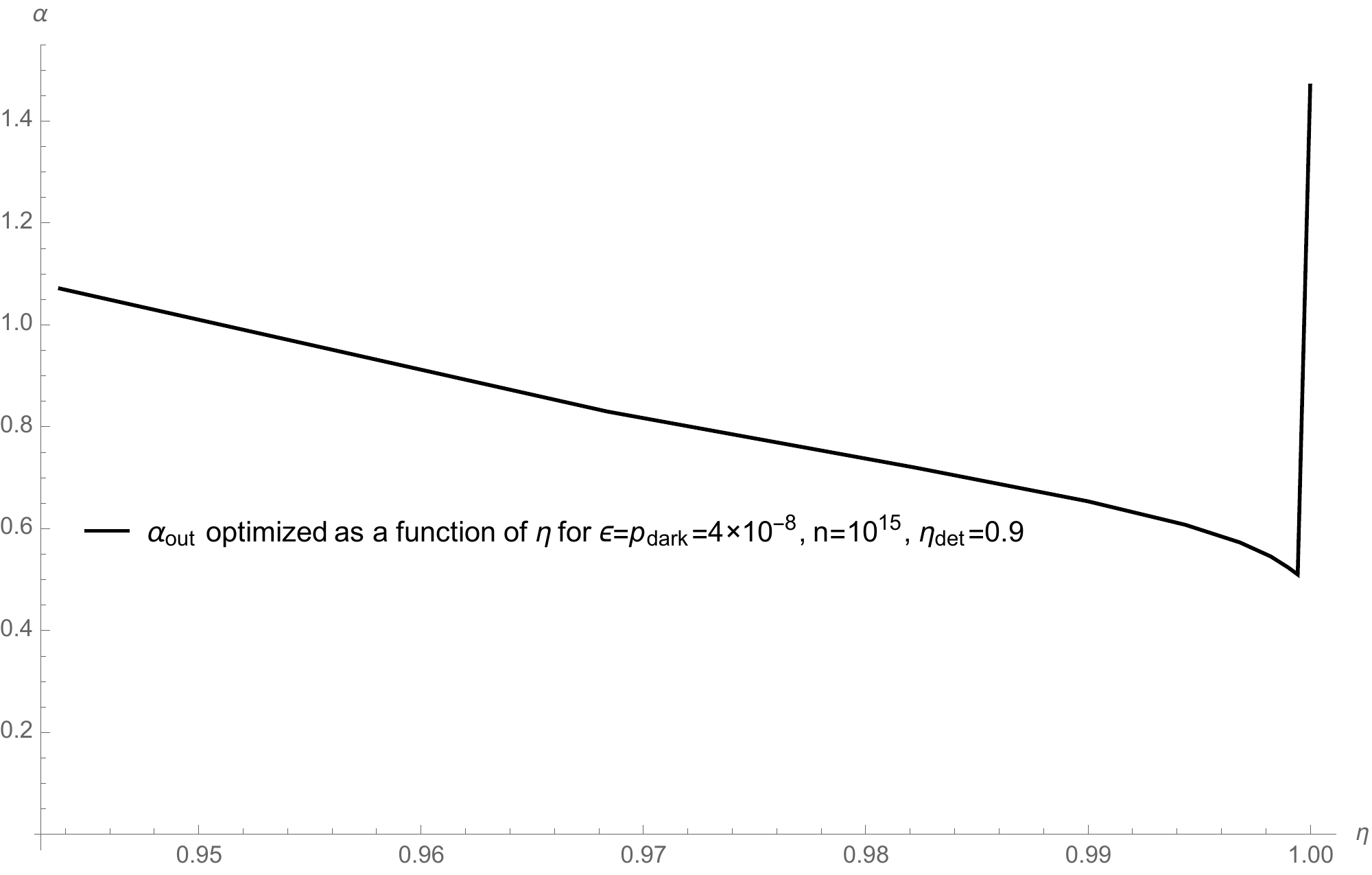}
\end{subfigure}
\caption{These figures depict $r$ on the left, and $\alpha$ on the right, optimized to minimize the information leakage, as functions of the transmissivity $\eta$ for fixed $n = 10^{15}$, $p_{dark}=4\times 10^{-8}$, and $\eta_{det}=0.9$, under the choice $s=0.001 n$. At $\eta=1$ the optimal values are $\alpha=1.47$ and  $r=1798$.}
\label{fig:opt_r_and_alpha}
\end{figure}
These parameters are 
challenging, but nevertheless seem achievable in {the} near future, especially 
taking into account that transmission distance is not crucial to the parties performing these protocols: the parties can bring the data close to each other before running the protocol at close proximity.


\section{Discussion}

We have proposed a quantum appointment scheduling protocol that requires only coherent states 
over {two modes} and basic linear optics operations over these two modes, {along with} classical processing. 
Our protocol 
shows 
that this important task can be 
realized in principle with such technology while providing a quantum advantage over any classical protocols 
in terms of information leakage.
%
The experimental parameters required to obtain a quantum 
advantage are challenging and have not yet been realized in the lab. However, we believe that these could 
be achieved in the near future, especially when taking into consideration that both parties can be in same 
lab, not far apart, for these private computations.

We found that the most limiting experimental parameter is the loss parameter~$\eta$. 
In Figure~\ref{fig:adv_vs_eta}, we plot the quantum advantage in terms of information leakage as a function of $\eta$.
We find that with constant loss, we have $\Omega (n)$ for the information leakage, with some constant information advantage ratio in a suitable parameter regime.


More generally, for $n$-bit inputs and $r$ rounds of interactions, for $r \leq n^{1/3}$, if $\eta \approx 1- \frac{1}{r}$, by slightly adapting the protocol we can get the information leakage as low as $\widetilde{O} (\frac{n}{r})$, leading to an asymptotic quantum advantage of  $\widetilde{O} (n^{2/3})$ vs. the classical $\Omega( n)$ lower bound. Our protocol in the ideal setting (for $\eta = 1$ and $p_{dark} =  0$) already achieves this asymptotic advantage, which is displayed in Figure~\ref{fig:ideal_qu_adv}.

Our work also opens up multiple interesting avenues of research into practical interactive communication. 
First, it will be interesting to see if the experimental 
parameters required to achieve a quantum advantage with our protocol can be achieved in the near 
future, and then whether the different components can be put together to obtain such a quantum advantage for 
the task of appointment scheduling. It will also be interesting to see how much further it is possible 
to improve practical protocols implementing the appointment scheduling task that we consider. 
We also hope that the tools that we develop in this work will serve to develop practical protocols 
with quantum advantage for other important distributed tasks.

\textbf{Acknowledgements.}

This research was supported in part by
 NSERC, Industry Canada and ARL CDQI program. The Institute for Quantum Computing and the Perimeter Institute for Theoretical Physics are supported in part by the Government of Canada and the Province of
Ontario.

\appendix
\section*{Appendix}

%
%

\section{Information Leakage}\label{leakage}

In this section we formally define the information leakage, and develop some properties which we will use to bound the information leakage of our appointment scheduling protocols. 
We  define quantum information complexity and discuss its link to privacy and some of its properties. We then specialize the discussion to the case of pure state protocols, further specialize to protocols with one-bit inputs, and then with no pre-shared entanglement.

\subsection{Quantum Information Complexity}

\subsubsection{Definition}

		\begin{figure}
		\begin{overpic}[width=1\textwidth]{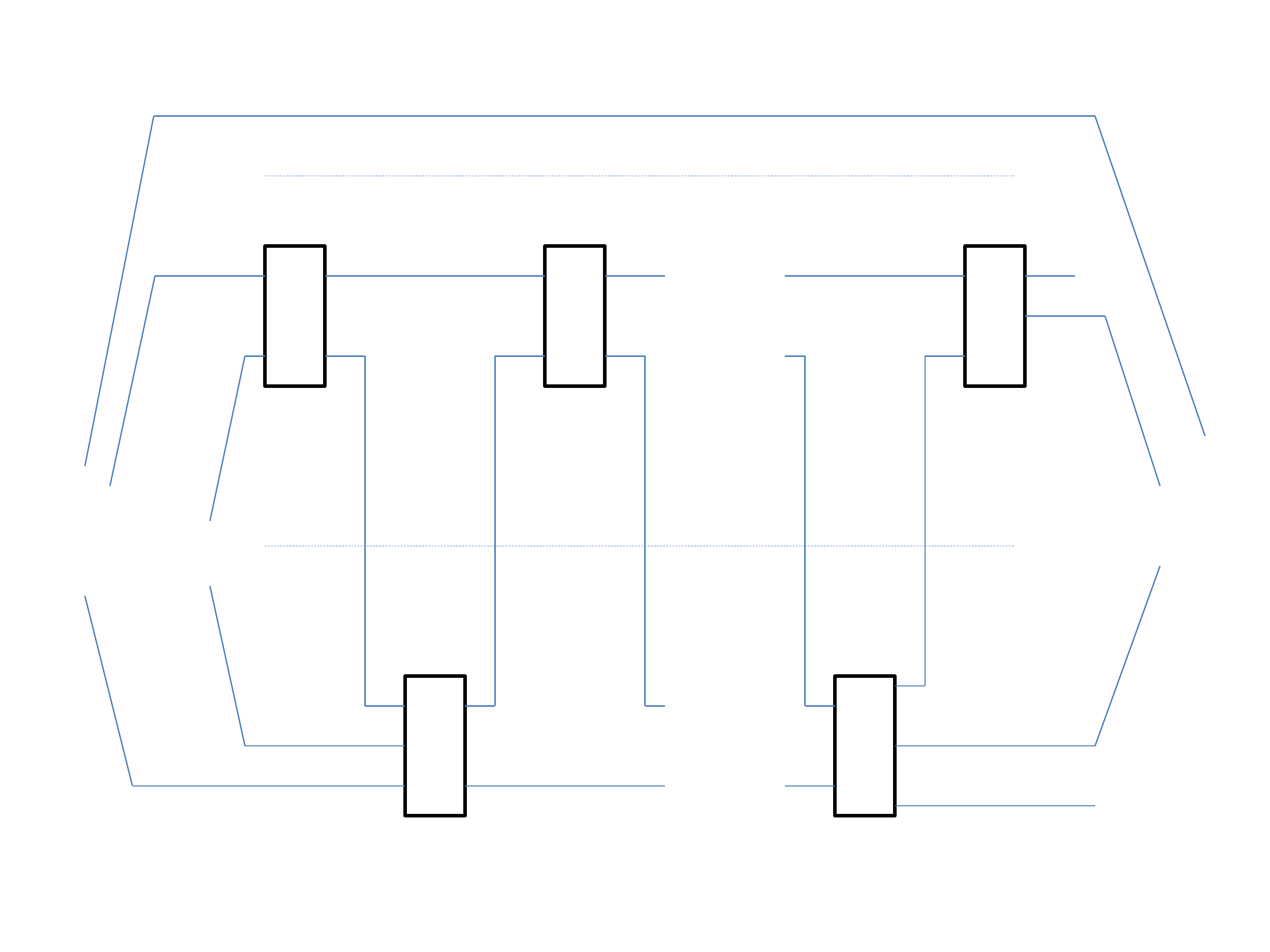}
		  \put(0,65){Ray}
		  \put(0,45){Alice}
		  \put(0,22){Bob}
		  \put(6,31){\footnotesize $\ket{\rho}$}
		  \put(15,66.5){\footnotesize$R$}
		  \put(15,54){\footnotesize$A_{in}$}
		  \put(15,13.5){\footnotesize$B_{in}$}
		  \put(15,44){\footnotesize$T_A^{in}$}
		  \put(15,19){\footnotesize$T_B^{in}$}
		  \put(22.1,49.5){\footnotesize$U_1$}
		  \put(15,31){\footnotesize{$\ket{\psi}$}}
		  \put(26.2,54){\footnotesize$A_1$}
		  \put(26.2,48){\footnotesize$C_1$}
		  \put(33,15.5){\footnotesize$U_2$}
		  \put(37.2,54){\footnotesize$A_2$}
		  \put(37.2,17.4){\footnotesize$C_2$}
		  \put(37.2,13.7){\footnotesize$B_2$}
		  \put(44.2,49.5){\footnotesize$U_3$}
		  \put(48.2,54){\footnotesize$A_3$}
		  \put(48.2,48){\footnotesize$C_3$}
		  \put(48.2,13.7){\footnotesize$B_3$}
		  \put(55,33){\footnotesize$\cdots$}
		  \put(59.5,54){\footnotesize$A_{M-1}$}
		  \put(59.5,48){\footnotesize$C_{M-1}$}
		  \put(58.5,13.7){\footnotesize$B_{M-1}$}
		  \put(66.2,15.5){\footnotesize$U_{M}$}
		  \put(71.5,54){\footnotesize$A_M$}
		  \put(73.5,23){\footnotesize$C_M$}
		  \put(73.5,16.4){\footnotesize$B_{out}$}
		  \put(73.5,11.7){\footnotesize$B^{\prime}$}
		  \put(77.3,49.5){\footnotesize$U_{f}$}
		  \put(81.3,54){\footnotesize$A^\prime$}
		  \put(81.3,50.3){\footnotesize$A_{out}$}
		  \put(92,33){\footnotesize$\Pi (\rho)$}
		\end{overpic}
		  \caption{Depiction of a quantum protocol in the interactive model, adapted from the long version of~\cite[Figure 1]{touchette2015quantum}. More details about the interactive model of quantum communication can be found there.}
		  \label{fig:int_mod}
		\end{figure}
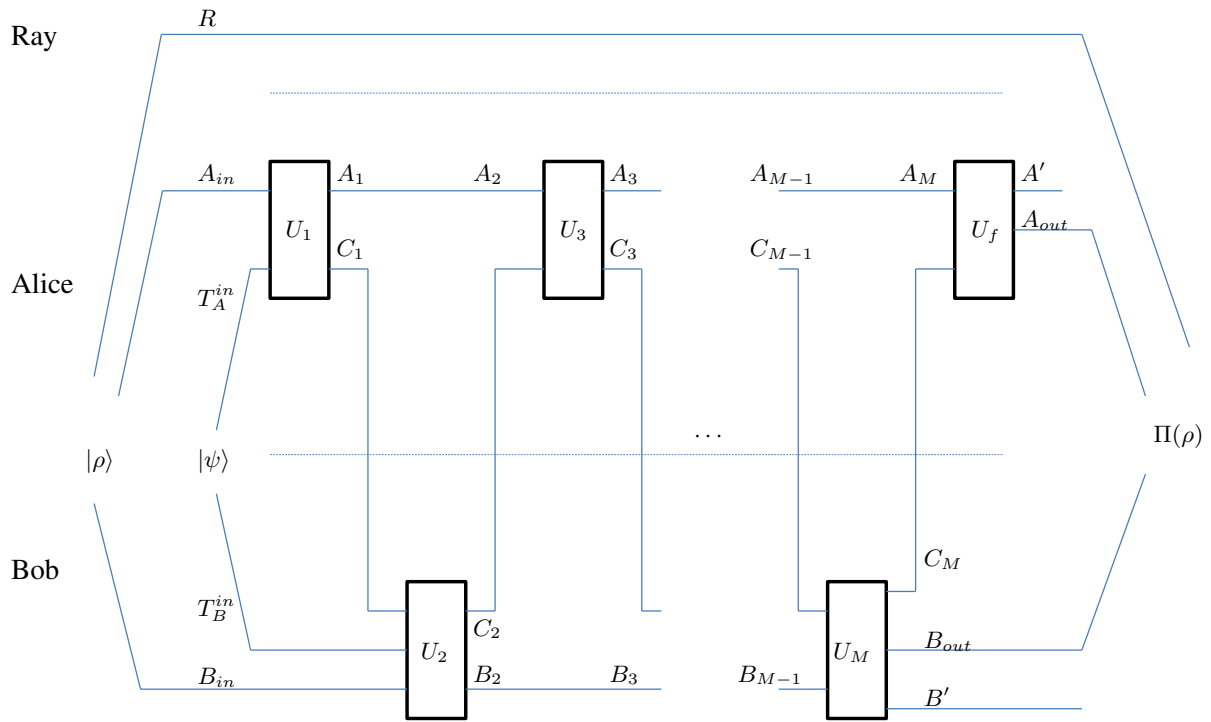

A quantity of interest in this work is the quantum information cost (or quantum information leakage, as written in the main text), as introduced in~\cite{touchette2015quantum}. We use an equivalent characterization given in~\cite{LauriereT:2016}. In quantum communication protocols, there is no clear notion of a transcript, so this definition quantifies how much information is exchanged in each round. In the sequel, we denote the von~Neumann entropy by $H$, and for a tripartite state $\rho^{ABC}$, we denote  the conditional entropy $H(A|B) = H(AB) - H(B)$ and the conditional quantum mutual information (CQMI) between $A$ and $B$ conditioned on $C$ by $I(A:B|C) = H(A | C)   - H(A | B,C)$.
Based on CQMI, the definition of quantum information cost of a protocol is as follows. The registers refer to those in Figure~\ref{fig:int_mod}.
\begin{definition}
\label{prelim:def:QICprotocol}
For a protocol $\Pi$ and an input distribution $\mu$, 
we define the \emph{quantum information cost} of the $i$th message of $\Pi$ on input distribution $\mu $ as
\begin{align*}
	\QIC_i (\Pi, \mu) &=  I(C_i; X | Y B_{i}) +   I(C_i; Y | X A_{i}),
\end{align*}
the \emph{quantum information cost} of $\Pi$ on input distribution $\mu $ as
\begin{align*}
	\QIC (\Pi, \mu) &= \sum_{i} \QIC_i (\Pi, \mu),
\end{align*}
and the (prior-free) \emph{quantum information cost} of $\Pi$ as
\begin{align*}
	\QIC (\Pi) &= \max_\mu \QIC (\Pi, \mu).
\end{align*}
\end{definition}

We discuss some of these properties as well as the connection to privacy in the next sections.
As discussed in Ref.~\cite{LauriereT:2016}, this is also well defined for classical protocols and then the definition agree with the one used in the classical literature~\cite{Barak:2010:CIC:1806689.1806701, Braverman2011, braverman2015interactive}.

Given a relation $T$, we can define the the quantum information complexity of $T$ by optimizing the quantum information cost over all protocols that compute~$T$.
We are interested in a prior-free notion of quantum information complexity, as introduced in~\cite{braverman2015near}.
\begin{definition}
\label{prelim:def:QICrelation}
Given a relation $T$ and an error parameter $\epsilon \geq 0$, we define the 
\emph{quantum information complexity} of $T$ at $\epsilon$ error as
\begin{align*}
	\QIC (T, \epsilon) &= \inf_\Pi \QIC (\Pi),
\end{align*}
in which the infimum is taken over all protocols that compute $T$ with error at most $\epsilon$ on all inputs.
\end{definition}

\subsubsection{Privacy}


One motivation to take $\QIC$ as a complexity measure for our quantum protocol is to have a fair comparison to classical notions. Both these notions characterize amortized communication complexity to solve multiple instances of the same problem in parallel, or equivalently, how much the messages arising from multiple copies of the protocol can be compressed. 
They also share many other important properties, and as such $\QIC$ is viewed as the quantum analogue to the classical notion of information leakage~\cite{braverman2015interactive}.
Moreover, these notions agree for classical protocols.

Another motivation was in terms of privacy concerns in a quantum honest-but-curious type of model, in which we want the parties to exchange the correct messages, but they might collect as much information as possible about each other's input. We want to bound how much information such parties might have at the end of protocols. However, we know from prior work that one must be really careful how to define quantum information, that unexpected behavior can arise~\cite{cleve1999quantum, baumeler2015quantum, chailloux2017information}. Depending on the situation, many different definitions have been put forward~\cite{JainRS09, klauck2002quantum, salvail2009power, touchette2015quantum, kerenidis2016information}. The link between many of these was studied in Ref.~\cite{LauriereT:2016}, and it was found that the notion of $\QIC$ we use is an upper bound on all of these. Hence, with this upper bound on the quantum information leakage, the difference between it and the classical bound will be smaller than for some other choices, but it will be robust.

\subsubsection{Properties}

We make use of many properties of conditional entropy and CQMI, among which the followings. Note that all $\log$ are base two.
\begin{lemma}
\label{prelim:lem:CQMIfacts}
	If $\rho = \rho^{ABCD} = \sum_{c} p(c) \kb{c}{c} \otimes \rho_c^{ABD} $ is a classical-quantum state with classical register $C$, then (Conditioning on a classical register is taking the average)
	$$
		H(A|CD)_{\rho} = \bbE_c\left[ H(A|D)_{\rho_c} \right],
	$$
	$$
		I(A:B|CD)_{\rho} = \bbE_c\left[ I(A:B|D)_{\rho_c} \right],
	$$
	and also (Dimension bound) 
	$$
		I(A:C|D)_{\rho} \leq \log \dim (C),
	$$
	$$
		I(A:B|D)_{\rho}  \leq 2 \log \dim (B),
	$$
	$$
		H(A)\leq \log \dim(A).
	$$
	If $\rho = \rho^{AB} = \kb{\psi}{\psi}^A \otimes \rho^B$ is pure on system $A$, then (Pure states have no entropy)
	$$
		H(A | B)_{\rho} = 0.
	$$
\end{lemma}

\begin{lemma}[Data processing inequality]
\label{lem:dpi}
For any quantum state $\rho^{ABC}$,
\begin{align*}
H(A|BC)_\rho\leq H(A|B)_\rho.
\end{align*}
\end{lemma}

\begin{lemma}[Isometric invariance]\label{isometry_entropy}
For any quantum state $\rho^A$ and any isometry $V \in \U(A,B)$,
\begin{align*}
H(A)_\rho=H(B)_{V \rho V^\dg}.
\end{align*}
\end{lemma}


We also make use of many properties of $\QIC$. The proof of the properties in the following lemma can be found in~\cite{touchette2015quantum, braverman2015near}.
\begin{lemma} \label{lem:prop_QIC}
Let $\nu$ be a distribution over input distributions $\mu$, and denote $\mu \sim \nu$ the random distribution $\mu$ over inputs picked with probability $\nu (\mu)$. Denote $\bar{\mu} = \bbE_{\mu \sim \nu}[\mu]$. Then for any protocol $\Pi$ (Concavity in input distribution)
\begin{align*}
	\bbE_{\mu \sim \nu} [\QIC (\Pi, \mu)] \leq \QIC (\Pi, \bar{\mu}).
\end{align*}
For any $p \in [0, 1]$ and any two input distributions $\mu_1$ and $\mu_2$ on $XY$, let $\mu = p \mu_1 + (1-p) \mu_2$. The following then holds for any $r$-message protocol $\Pi$ (Quasi-convexity in input distribution)
\begin{align*}
	\QIC (\Pi, \mu) \leq p \QIC (\Pi, \mu_1) + (1 - p) \QIC (\Pi, \mu_2) + 2r \;  h(p).
\end{align*}
For any two protocols $\Pi_1$ and $\Pi_2$, the following holds for the protocol $\Pi_1 \otimes \Pi_2$, running them in parallel, and for any joint input $\mu_{12}$ (Subadditivity)
\begin{align*}
	\QIC (\Pi_1 \otimes \Pi_2, \mu_{12}) \leq  \QIC (\Pi_1, \mu_1) + \QIC (\Pi_2, \mu_2),
\end{align*}
with $\mu_1$ and $\mu_2$ the marginal of $\mu_{12}$ for the input to protocols $\Pi_1$ and $\Pi_2$, respectively.
\end{lemma}


The following lemmata
 show that when running a protocol as a subroutine, classical side-information can be conditioned on, and quantum side-information can be safely discarded without increasing quantum information cost. We first introduce some notation.

\begin{definition}
\label{def:condqic}
Let $\rho_\mu^{A_{in} B_{in} O_A O_B}$ be a state with purification of the form 
\begin{align*}
	\ket{\rho_\mu}^{A_{in} B_{in} O_A O_B R} = \sum_o {\sqrt{p_O (o)}} \ket{o}^{O_A} \ket{o}^{O_B} \ket{\rho_{\mu_o}}^{A_{in} B_{in} R^o},
\end{align*}
for some distribution $p_O$ and input distributions $\mu_o$ satisfying $\mu = \sum_o p_O (o) \mu_o$, and where purification register $R^o = R_X^o R_Y^o$. Also let $\Pi$ be a protocol acting on input registers $A_{in} B_{in} = XY$. Then we define
the \emph{quantum information cost} of $\Pi$ on $\mu | O$ as
\begin{align*}
	\QIC (\Pi, \mu| O) = \sum_{i} I(C_i ; X | Y B_i O_B) +  I(C_i ; Y | X A_i O_A).
\end{align*}
\end{definition}

\begin{definition}
\label{def:sideqic}
Let $\Pi$ be a  protocol acting on input registers $A_{in} B_{in} = XY$, let $\rho_{\mu}^{A_{in} B_{in}}$ and $\sigma^{A_{in}^\prime B_{in}^\prime}$ be states, with $A_{in}^\prime = A_{in}  \tilde{A}$ and $B_{in}^\prime = B_{in}  \tilde{B}$ for some arbitrary finite dimensional registers $\tilde{A}, \tilde{B}$, and such that $\Tr{\tilde{A} \tilde{B}}{\sigma} = \rho_{\mu}$.
Then we define
the \emph{quantum information cost} of $\Pi$ on $\mu$ with \emph{side information} $\sigma$ as
\begin{align*}
	\QIC (\Pi, \mu| \sigma) = \sum_{i} I(C_i ; X | Y B_i \tilde{B}) +  I(C_i ; Y | X A_i \tilde{A}).
\end{align*}
\end{definition}

The next lemma follows from definitions.

\begin{lemma}
In the setting of Definition~\ref{def:condqic}, define $\sigma^{A_{in}^\prime B_{in}^\prime} = \rho^{A_{in} B_{in} O_A O_B}$, with $A_{in}^\prime = A_{in}  O_A$ and $B_{in}^\prime = B_{in}  O_B$, and such that $\Tr{O_{A} O_{B}}{\rho^{A_{in} B_{in} O_A O_B}} = \rho_{\mu}$. Then 
\begin{align*}
	\QIC (\Pi, \mu | \sigma) & = \QIC (\Pi, \mu | O).
\end{align*}

\end{lemma}

The next lemma follows directly from the definition and because conditioning on a classical register is taking the average (Lemma~\ref{prelim:lem:CQMIfacts}), noting that $O_A$ and $O_B$ are perfectly correlated and are classical once one of them is traced out.

\begin{lemma}[Conditioning on a common variable is taking average] \label{lem:qic_cond_av}
In the setting of Definition~\ref{def:condqic}, 
\begin{align*}
	\QIC(\Pi, \mu | O) = \sum_o p_O (o) \QIC (\Pi, \mu_o).
\end{align*}

\end{lemma}

The next lemma follows from subadditivity (a stronger version than the one stated here, which is proved in Ref.~\cite{braverman2015near} and also holds for quantum state inputs and is equivalent in the case of classical inputs with side information) and the fact that we can implement an identity channel with the trivial protocol that does not communicate at all.

\begin{lemma}[Safe discarding of side-information]\label{QICdiscard}

Let $\Pi$ be a  protocol acting on input registers $A_{in} B_{in} = XY$, $\rho_{\mu}^{A_{in} B_{in}}$ and $\sigma^{A_{in}^\prime B_{in}^\prime}$ be states, with $A_{in}^\prime = A_{in} \otimes \tilde{A}$ and $B_{in}^\prime = B_{in} \otimes \tilde{B}$ for some arbitrary finite dimensional registers $\tilde{A}, \tilde{B}$, and such that $\Tr{\tilde{A} \tilde{B}}{\sigma} = \rho_\mu$. Then 
\begin{align*}
	\QIC (\Pi, \mu | \sigma) & \leq \QIC (\Pi, \mu).
\end{align*}

\end{lemma}

We will consider inconclusive protocols which compute a given relation but might also return output ``inconclusive'' with some probability.
A particularly important class of such protocols in our setting are recursively defined as follows, which generalizes our $AND$ protocol $\Pi_A$.

\begin{definition}
\label{def:recursive}

Given a protocol $\widetilde{\Pi}$ which can be inconclusive, and such that Alice and Bob always agree on whether a run was inconclusive or not, we recursively define $\Pi (\widetilde{\Pi})$ as follows:

Protocol $\Pi (\tilde{\Pi})$:
\begin{enumerate}
\item Run Protocol $\widetilde{\Pi}$.
\item If $\widetilde{\Pi}$ returns an output, return this output.
\item Else, if $\widetilde{\Pi}$ is inconclusive, rerun $\Pi (\widetilde{\Pi})$
\end{enumerate}

\end{definition}

In particular, we will be interested in protocols $\widetilde{\Pi}$ that  have the same probability $p$ to be inconclusive for all their inputs.
The following bound holds for such protocols.

\begin{lemma}[Rerunning an inconclusive protocol]\label{eq:rerun}

Let $\widetilde{\Pi}$ be a protocol that has uniform probability $p < 1$ of being inconclusive for all input pairs $(x, y)$, and let $\Pi (\widetilde{\Pi})$ be as in Definition~\ref{def:recursive}. Then for any input distribution $\mu$ it holds that

\begin{align*}
	\QIC (\Pi (\widetilde{\Pi}), \mu ) & \leq \frac{1}{1-p} \QIC (\widetilde{\Pi}, \mu).
\end{align*}

\end{lemma}

\begin{proof}
In the setting of Definition~\ref{def:sideqic}, let $\widetilde{A} = O_A A_{left}$ and $\widetilde{B} = O_B B_{left}$ with $O_A$, $O_B$ indicator variables (as in the setting of Definition~\ref{def:condqic}) for whether the run was inconclusive and $A_{left}$, $B_{left}$ leftover registers when a run is inconclusive. Then, if $O_A$ and $O_B$ are set to zero, the protocol ends and the only cost incurred is that of $\widetilde{\Pi}$, and else $\Pi (\widetilde{\Pi})$ is rerun entirely. Notice that since the probability $p$ that the run is inconclusive is uniform for all inputs $(x, y)$, it holds that $\mu|O_A = 1$ is distributed as $\mu$. Denoting $\sigma_1$ the state on $XY A_{left} B_{left}$ when $O_A = 1$, we get

\begin{align*}
	\QIC (\Pi (\widetilde{\Pi}), \mu ) & =  \QIC (\widetilde{\Pi}, \mu) + \Pr[O_A = 1] \QIC (\Pi (\widetilde{\Pi}), \mu | \sigma_1 ) \\
					& \leq \QIC (\widetilde{\Pi}, \mu) + p \QIC (\Pi (\widetilde{\Pi}), \mu ),
\end{align*}

in which the inequality follows by safe discarding of quantum side-information (Lemma~\ref{QICdiscard}). The result follows by rearranging terms.

\end{proof}

\subsection{Quantum Information Cost of Pure State Protocols}

\subsubsection{General Pure State Protocols}

We consider protocols for which, conditional on fixed inputs $x, y$, pure states $\ket{\phi_i^{x, y}}^C$ are exchanged. For any such protocol $\Pi$ and any input distribution $\mu$, we can then rewrite for the $i$th term of the quantum information cost:

\begin{align}
\QIC_i (\Pi, \mu) & = I(X; C_i | Y B_i) + I(Y; C_i | X A_i) \\
	& = H(C_i | Y B_i) - H(C_i | X Y B_i) + H(C_i | X A_i) - H(C_i | X Y A_i) \\
	& = H(C_i | Y B_i) + H(C_i | X A_i)  \\
	& = \bE_{y \sim \mu_Y} [ H(C_i B_i | Y = y) - H(B_i | Y = y) ] \label{eq:purestateinfo} \\
	& \quad \quad 
+ \bE_{x \sim \mu_X} [ H(C_i A_i | X = x) - H(A_i | X = x)  ]. \nonumber
\end{align}

The third equality follows by expanding over conditioning registers $XY$ (Lemma~\ref{prelim:lem:CQMIfacts}) and then using that the messages in register $C$ are pure states (Lemma~\ref{prelim:lem:CQMIfacts}).
We denote by $\mu_X$ and $\mu_Y$ the marginals of $\mu$  on $X$ and $Y$, respectively.
This is a similar to the form of $\QIC$ for memoryless protocols studied in Ref.~\cite{chailloux2017information}.

\subsubsection{Protocols with One-bit Inputs}
\label{sec:ANDctny}

For computing bitwise $AND$, we can further use the fact that $X$ and $Y$ are single bits. 
We will be interested in distributions $\mu_w$ on $XY$ with very small mass $\mu_w (1, 1) =w$, 
hence we consider a distribution that is close to the extreme case $\mu_0^w$ given by
$\mu_0^w (1,1) = 0$ and $\mu_0^w (x,y) = \frac{1}{1-w} \mu_w (x,y)$ for $(x, y) \not= (1, 1)$. We handle the non-zero but small mass on $(1, 1)$ 
by quasi-convexity in the input distribution. Note that $\QIC(\Pi, \mu_1) = 0$ for $\mu_1$ such that $\mu_1 (1, 1) = 1$.
%
%
%
The following bound holds for any $M$ message protocol and was derived in Ref.~\cite{braverman2015near}.

\begin{lemma}[Continuity for low mass protocols]
\label{lem:ctny-low-mass}
For any $M$ message
protocol $\Pi$ and any input distribution $\mu_w$ as above, 
\begin{align*}
\QIC (\Pi, \mu_w)  &\leq \QIC (\Pi, \mu_{0}^w) + 2 M \; h(w).
\end{align*}
\end{lemma}

Hence, the set of input distributions $\mu_0$ with no mass on $(1, 1)$ will play a special role. We have the following definition.
\begin{definition}
For any protocol $\Pi$ for $AND$, we define 
\begin{align*}
\QIC_0 (\Pi)  &= \max_{\mu_0 : \mu_0 (1, 1)  =0} \QIC (\Pi, \mu_{0}).
\end{align*}
\end{definition}

\subsubsection{Protocols with no Pre-shared Entanglement}
\label{sec:protocol_no_ent}

We will be considering protocols with no pre-shared entanglement. For such protocols, the following remark can be seen to hold by an inductive argument.

\begin{remark}
For any protocol with pure state messages and no pre-shared entanglement, it holds that registers $A_i$, $B_i$, and $C_i$ are all pure, separable states conditional on $X,Y$.
\end{remark}

For protocols with one-bit inputs, we are thus left with computing the different entropies in $\QIC$, all over states corresponding to an ensemble of two pure states each with some \textit{a priori} distribution. Computing such entropies can be reduced to a function of the overlap $F$ between the two pure states and the probability $p$ of having the first of these states.

\begin{lemma}\label{ranktwofid}
For any two pure states $\ket{\psi}^A$ and $\ket{\phi}^A$ with overlap $F (\ket{\psi}, \ket{\phi}) = |\bk{\psi}{\phi}| = F$ and probability $p$ of having $\ket{\psi}^A$, the entropy $H(A)_\rho$ of the average state $\rho^A = p \kb{\psi}{\psi} + (1-p) \kb{\phi}{\phi}$ can be computed as a function of $F$ and p:

\begin{align*}
	H(A)_\rho & = h \left(\frac{1}{2} - \frac{1}{2} \sqrt{1 - 4p(1-p)(1-F^2)}\right) \\
			& \leq h \left(\frac{1}{2} (1-F)\right),
\end{align*}
with the binary entropy $h(\epsilon) = - \epsilon \log \epsilon - (1 - \epsilon) \log (1 - \epsilon)$ for any $\epsilon \in [0, 1]$.
\end{lemma}

The inequality follows 
since, for any fixed $F \in [0, 1)$, the binary entropy $h(\frac{1}{2} - \frac{1}{2} \sqrt{1 - 4p(1-p)(1-F^2)})$ is maximized at $p=\frac{1}{2}$.

\subsection{Appointment Scheduling and Disjointness}

The main problem that we study is the appointment scheduling problem.
In this problem, Alice and Bob each have a calendar and know for each date whether they are available or not to schedule an appointment.
The goal is for Alice and Bob to determine a date where they are both available.
In case where  at least one such date exists, we require that both Alice and Bob output the same date.
In case when no such date exists, we require that they both output that their calendars are non-intersecting.
More formally, we represent $n$-date calendars by input strings $x,y \in \ZO^n$, with a one at position $i$ indicating availability and a zero non-availability.
The goal of Alice and Bob is to both output the same date $i \in [n]$, with $[n] = \{ 1, 2, \ldots, n \}$, such that $x_i = y_i = 1$ if such a date exists, or else both output $``\emptyset"$ saying that their calendars are non-intersecting.

A closely related problem is the disjointness function for $n$-bit inputs, defined as: for all $x,y \in \ZO^n$,
\begin{align*}
\DISJ_n(x,y) = \neg \left(\OR_{i \in [n]} (x_i \, \AND \, y_i)\right). 
\end{align*}

In fact, any protocol solving the appointment scheduling problem can be converted into a protocol with output on both sides solving the disjointness function without changing the communication or the probability of error as follows: 
\begin{enumerate}
\item Alice and Bob run the protocol for appointment scheduling and each get an output.
\item If the output is that no intersection exists, they output that the sets are disjoint.
\item Else, if the output is some date of intersection, they output that the sets are not disjoint.
\end{enumerate}

The following bounds is proven in Ref.~\cite{braverman2013information} for computing the  Disjointness function with zero-error.
\begin{theorem}
Any zero-error classical protocol $\Pi_C$ for computing the Disjointness function on $n$-bit inputs has information leakage satisfying $IC (\Pi_C) \geq 0.48n$.
\end{theorem}

The result was extended to non-zero error protocols in Ref.~\cite{dagan2016trading}, who obtained the following bounds.
\begin{theorem}
Any classical protocol $\Pi_C$ for computing the Disjointness function on $n$-bit inputs with error at most $\epsilon>0$ has information leakage satisfying $IC (\Pi_C) \geq (0.48 - 16 \; h(\sqrt{\epsilon}))n$.
\end{theorem}
It is also shown in Ref.~\cite{dagan2016trading} that the $16 \; h(\sqrt{\epsilon})$ term can be replaced by the potentially much smaller $O(h(\epsilon))$ for the Disjointness function. However, the constant is left unexplicit in that case.

\section{Information Leakage Analysis for our Appointment Scheduling protocol: Ideal setting}\label{leakage_ideal}

%


In this section we analyze the ideal setting information leakage of the appointment scheduling protocol developed in the main text. First we prove Lemma~\ref{QIC_general}, which holds for any appointment scheduling protocol using an arbitrary zero-error protocol $\Pi_A$ for AND. We then use this lemma to bound the information leakage of our protocol.

\subsection{$\mathrm{QIC}$ for $\Pi_D$ with a generic $\Pi_A$}

\begin{lemma}\label{QIC_general}
Given a zero-error protocol $\Pi_A$ for $AND$ built from protocol $\widetilde{\Pi}_A$ as in Definition~\ref{def:recursive} that has uniform probability $p$ of being inconclusive, the protocol $\Pi_D$ described in Section~\ref{coh_prot} satisfies
\begin{align*}
\QIC (\Pi_D) & \leq s + \log s + 1 + \frac{n}{ 1 - p} \max \Bigg[\frac{2(2r+1)}{n},  \\
		&\hspace{1.5in} \QIC_0 (\widetilde{\Pi}_A) + 2(2r +1) \; h\left(\frac{2  \ln n }{s } + \frac{1}{n}\right)\Bigg]
\end{align*}
for all values of $n$ and $s$ satisfying $n\geq 4$ and $8\ln(n)\leq s \leq n$.
\end{lemma}

Note that for zero error, running protocol $\Pi_D$ in terms of $\Pi_A$ and $\widetilde{\Pi}_A$ as in the main text or in terms of $\Pi_A (\widetilde{\Pi}_A)$ as per Definition~\ref{def:recursive} leads to the same $\QIC$. This follows since the extra information sent 
in step~3 of   $\Pi_A$ 
can be computed locally, because for zero-error protocols 
this extra information is sent only when both inputs are equal to ``1''.
For the proof below, we thus consider $\Pi_A (\widetilde{\Pi}_A)$ as per Definition~\ref{def:recursive}, and virtually avoid the transmission of this extra information in the ideal case.

We want to show that $\QIC (\Pi_D) = \max_\mu \QIC (\Pi_D, \mu)$ is small. Fix any such~$\mu$. 
The bitwise $AND$ protocol is only guaranteed to have low information leakage for low probability of intersection.
The classical subsampling part serves to ensure that, in the case where no intersecting date is found, we can update our knowledge about $\mu$, such that with high probability, it has few intersections in the remaining dates. 

Hence, an average date has low probability of being an intersecting date, and we can run our bitwise $AND$ protocol on each date without incurring too high an information leakage. We formalize this as follows.

\begin{proof}
We first upper bound the information leakage in the classical subsampling part by the total amount of communication arising in that phase (Lemma~\ref{prelim:lem:CQMIfacts} with classical communication register), $s$ for the compared bit values and $\log s  +1$ to indicate whether there is coincidence and if so, the position of coincidence within the subsampled set.
Let $S_A$ be an indicator random variable for whether an intersecting date was found while subsampling.
We first use the results that conditioning on a common variable is taking average (Lemma~\ref{lem:qic_cond_av}) and that we can safely discard side-information (Lemma~\ref{QICdiscard}) 
to notice that the information cost corresponding to running the bitwise $AND$ protocol $\Pi_A$  in parallel is upper bounded by $\QIC (\Pi_A^{\otimes n}, \nu)$, in which $\nu$ is the distribution resulting from conditioning the distribution $\mu$ on the observation $S_A=0$, denoted  $\mu | S_A = 0$. Denoting $\nu_i$ the marginal of $\nu$ in the $i$th date, we can then bound

\begin{align}
\QIC (\Pi_D, \mu) & \leq s + \log s + 1 + \Pr[S_A = 0] \QIC (\Pi_A^{\otimes n }, \nu) \\
			& \leq s + \log s + 1 +  \Pr [S_A = 0] \sum_{i \in [n]}  \QIC (\Pi_A,  \nu_i) \\
			& \leq s + \log s + 1 +  \Pr [S_A = 0] n \QIC( \Pi_A, \frac{1}{n} \sum_{i \in [n]} \nu_i) \\
			& \leq s + \log s + 1 +  \Pr [S_A = 0] \frac{n}{1 - p} \QIC (\tilde{\Pi}_A, \frac{1}{n} \sum_{i \in [n]} \nu_i), 
\end{align}
in which the second inequality follows from subadditivity (Lemma~\ref{lem:prop_QIC}), the third from concavity in the input distribution (Lemma~\ref{lem:prop_QIC}) and the fourth by  rerunning an inconclusive protocol (Lemma~\ref{eq:rerun}).

We consider two cases, either $\Pr[S_A = 0] \leq 1/n$ or $\Pr[S_A = 0] > 1/n$.

If $\Pr[S_A = 0] \leq 1/n$, then we can use a dimension bound (Lemma~\ref{prelim:lem:CQMIfacts} on each classical input register) to get, for  the $2r+1$ messages in $\tilde{\Pi}_A$,

\begin{align}
 \Pr [S_A = 0]  \QIC (\tilde{\Pi}_A, \frac{1}{n} \sum_{i \in [n]} \nu_i) \leq  \frac{2(2r+1)}{n} ,
\end{align}
which completes the proof in this case.

If $\Pr[S_A = 0] > 1/n$, we need only to show that the classical subsampling stage ensures the inequality
\begin{align}\label{inequality}
\frac{1}{n} \sum_{i \in [n]} \nu_i(1,1)\leq \frac{2 \ln(n)}{s}+\frac{1}{n}\leq 1/2
\end{align}
for all values of $n$ and $s$ satisfying $n\geq 4$ and $8 \ln (n)\leq s \leq n$. An application of the continuity bound for low mass protocols (Lemma~\ref{lem:ctny-low-mass}) will complete the proof in this case. The second inequality in \eqref{inequality} is straightforward. For the first, let $N(X,Y)$ be a random variable outputting the number of intersecting dates of $(x,y)$. Note that
\begin{align}
\sum_{i \in [n]} \nu_i(1,1)&=\bE_{\nu} N(X, Y)=\bE_{\mu|S_A=0} N(X, Y)\nonumber\\
							&=\sum_{1 \leq d \leq n} \Pr[N(X,Y) = d | S_A = 0] \; d
\end{align}
and
\begin{align}
\Pr[N(X,Y) = d | S_A = 0] & = \frac{\Pr[N(X,Y) = d ]}{\Pr [S_A = 0]} \cdot \Pr[S_A = 0 | N(X,Y) = d ] \nonumber\\
				& \leq n \Pr[N(X,Y) = d ]  (1 - d/n)^{s} \nonumber\\
				& \leq n \Pr[N(X,Y) = d ]  \exp (-ds/n), \label{expbound}
\end{align}
where the first inequality follows from the case assumption $\Pr[S_A=0]\geq 1/n$ and $\Pr[S_A = 0 | N(X,Y) = d ]\leq (1 - d/n)^{s}$.
Thus,
\begin{align}
\sum_{i \in [n]} \nu_i(1,1)&=\sum_{1 \leq d \leq \cl{\frac{2 n \ln(n)}{s}}} \Pr[N(X,Y) = d | S_A = 0] \; d\nonumber\\
&+\sum_{\cl{\frac{2 n \ln(n)}{s}} < d \leq n} \Pr[N(X,Y) = d | S_A = 0] \; d\nonumber\\
&\leq \sum_{1 \leq d \leq \cl{\frac{2 n \ln(n)}{s}}} \Pr[N(X,Y) = d | S_A = 0] \; d \nonumber\\
&+\sum_{\cl{\frac{2 n \ln(n)}{s}} < d \leq n} n \Pr[N(X,Y) = d ] \exp (-ds/n) n\nonumber\\
&\leq \sum_{1 \leq d \leq \cl{\frac{2 n \ln(n)}{s}}} \Pr[N(X,Y) = d | S_A = 0] \; d \nonumber\\
&+ 1 \nonumber\\
&\leq \frac{2n \ln(n)}{s}+2. \nonumber
\end{align}
The first inequality follows from \eqref{expbound} and the upper bound $d\leq n$ for all $d$ in 
the range of the second sum. The second inequality follows from $\exp(-ds/n)\leq 1/n^2$ 
(which results from $2\ln(n) n / s < d$) and the fact that the sum over all $d$ of $\Pr[N(X, Y) = d]$ is 
at most $1$. The third inequality follows from  the fact that the sum is upper bounded by a convex 
combination of $1 \leq d\leq \cl{\frac{2 n \ln(n)}{s}}$, which is upper bounded by the largest 
term $\cl{\frac{2 n \ln(n)}{s}} \leq \frac{2 n \ln(n)}{s} +1$. This completes the proof of the 
inequality \eqref{inequality}, and  the proof of Lemma~\ref{QIC_general} follows from the continuity for low mass protocols, Lemma~\ref{lem:ctny-low-mass}.
\end{proof}

\subsection{$\mathrm{QIC}$ for subroutine $AND$ protocol}\label{QIC_subroutine}

Here we complete the proof of Theorem~\ref{QIC_general_ideal} by bounding the information cost $\QIC_0(\widetilde{\Pi}_A) = \max_{\mu_0: \mu_0 (1,1 ) = 0} \QIC(\widetilde{\Pi}_A, \mu_0) $ of our bitwise-$AND$ protocol $\Pi_A$ in the ideal setting.

The following lemma, along with Lemma~\ref{QIC_general} and the fact that the uniform probability $p$ of $\widetilde{\Pi}_A$ to be inconclusive is $p = \exp (- \abs{\alpha}^2)$, completes the proof of Theorem~\ref{QIC_general_ideal}.
\begin{lemma}
\begin{align*}
\QIC_0 (\widetilde{\Pi}_A) \leq h\left(\frac{1}{2} (1 - F(r, \alpha))\right),
\end{align*}
in which
\begin{align*}
F(r, \alpha) = \exp \left[-r\abs{\alpha}^2 \left[1 - \cos \left(\frac{\pi}{2r}\right)\right] \right].
\end{align*}
\end{lemma}

Now we prove this lemma. 
First, recall the evolution of the protocol on the different inputs.
In order to assess information leakage, Bob makes all his operations isometric.
Here, we explicitly record  in parentheses the state representing Bob's memory. Note that in practice, Bob can generate these states on the fly and need not keep any quantum memory.

\begin{samepage}
\textbf{Evolution of $\widetilde{\Pi}_A$ for different inputs:}

\begin{flushleft}
\textbf{On (0, 0):}
\begin{tabular}{l l l l}
$\ket{\alpha, 0} (\ket{\alpha, 0}^{\otimes r})$ & $\rightarrow_A \ket{\alpha, 0} (\ket{\alpha, 0}^{\otimes r})$ &$\rightarrow_B \ket{\alpha, 0} (\ket{\alpha, 0}^{\otimes r})$ & $\rightarrow_A \cdots $
\end{tabular}\\
\textbf{On (0,1):}
\begin{tabular}{l l l l}
$\ket{\alpha, 0} (\ket{\alpha, 0}^{\otimes r})$ & $\rightarrow_A \ket{\alpha, 0} (\ket{\alpha, 0}^{\otimes r})$ & $\rightarrow_B \ket{\alpha, 0} (\ket{\alpha, 0}^{\otimes r})$ & $\rightarrow_A \cdots $
\end{tabular}\\
\textbf{On (1, 0): }
\begin{tabular}{l l}
\hspace{-.2in}$\ket{\alpha, 0} (\ket{\alpha, 0}^{\otimes r})$ & \\
\hspace{-.2in}$\rightarrow_A \ket{\cos (\theta) \alpha, \sin (\theta ) \alpha} (\ket{\alpha, 0}^{\otimes r})$
&\hspace{-1.5in} $\rightarrow_B \ket{\alpha, 0} (\ket{\cos (\theta) \alpha, \sin (\theta ) \alpha} \ket{\alpha, 0}^{\otimes r-1})$\\
\hspace{-.2in}$\rightarrow_A \ket{\cos (\theta) \alpha, \sin (\theta ) \alpha} (\ket{\cos (\theta) \alpha, \sin (\theta ) \alpha} \ket{\alpha, 0}^{\otimes r-1})$ & \\ 
& \hspace{-1.5in}$\rightarrow_B \ket{\alpha, 0} (\ket{\cos (\theta) \alpha, \sin (\theta ) \alpha}^{\otimes 2} \ket{\alpha, 0}^{\otimes r-2})$ \\
\hspace{-.2in}$\rightarrow_A \ket{\cos (\theta) \alpha, \sin (\theta ) \alpha} (\ket{\cos (\theta) \alpha, \sin (\theta ) \alpha}^{\otimes 2} \ket{\alpha, 0}^{\otimes r-2})$
& \\
& \hspace{-1.5in} $\rightarrow_B \ket{\alpha, 0} (\ket{\cos (\theta) \alpha, \sin (\theta ) \alpha}^{\otimes 3} \ket{\alpha, 0}^{\otimes r-3})$ \\
$\vdots$ & $\vdots$ \\
\hspace{-.2in}$\rightarrow_A \ket{\cos (\theta) \alpha, \sin (\theta ) \alpha} (\ket{\cos (\theta) \alpha, \sin (\theta ) \alpha}^{\otimes r-1} \ket{\alpha, 0}^{\otimes 1})$
& \\
& \hspace{-1.5in}$\rightarrow_B \ket{\alpha, 0} (\ket{\cos (\theta) \alpha, \sin (\theta ) \alpha}^{\otimes r} )$\\
\hspace{-.2in}$= \ket{\alpha, 0} (\ket{\cos (\theta) \alpha, \sin (\theta ) \alpha}^{\otimes r} )$ &
\end{tabular}\\
\textbf{On (1, 1):}\\
\begin{tabular}{l l}
$\ket{\alpha, 0} (\ket{\alpha, 0}^{\otimes r})$ & \\
$\rightarrow_A \ket{\cos (\theta) \alpha, \sin (\theta ) \alpha} (\ket{\alpha, 0}^{\otimes r})$
& $\rightarrow_B \ket{\cos (\theta) \alpha, \sin (\theta ) \alpha} (\ket{\alpha, 0}^{\otimes r})$ \\
$\rightarrow_A \ket{\cos (2 \theta) \alpha, \sin (2 \theta ) \alpha} (\ket{\alpha, 0}^{\otimes r})$
& $\rightarrow_B \ket{\cos (2 \theta) \alpha, \sin (2 \theta ) \alpha} (\ket{\alpha, 0}^{\otimes r})$ \\
\vdots & \vdots \\
$\rightarrow_A \ket{\cos (r \theta) \alpha, \sin (r \theta ) \alpha} (\ket{\alpha, 0}^{\otimes r})$
& $\rightarrow_B \ket{\cos (\frac{\pi}{2}) \alpha , \sin (\frac{\pi}{2}) \alpha} (\ket{\alpha, 0}^{\otimes r})$ \\
$= \ket{0, \alpha} (\ket{ \alpha, 0}^{\otimes r} )$.
\end{tabular}
\end{flushleft}
\end{samepage}

We now upper bound $\QIC_0(\widetilde{\Pi}_A)$. 

We first handle the last message.
Under $\mu_0$, this last message is $0$ with probability $1- e^{-\abs{\alpha}^2}$ and inconclusive with probability $e^{-\abs{\alpha}^2}$. This holds independently of the input $(x, y)$ and of the content of $B_{2r+1}$. Thus, $\QIC_{2r+1} (\widetilde{\Pi}_A, \mu_0) = 0$.
This is the only non-pure message in the protocol.

The first $2r$ messages are pure, so we start from (\ref{eq:purestateinfo}), for which we have a trivial $A_i$ register here.
We first note that under $\mu_0$, if $X=1$ then $Y=0$, so $H(C_i | X =1) = 0$ because register $C_i$ is then in a pure state (Lemma~\ref{prelim:lem:CQMIfacts}). Similarly, if $Y=1$ then $X=0$ and $H(C_i | B_i, Y =1) = 0$ because register $C_i$ is then in a pure state.
Further,  on $X=0$, the pure state messages in $C_i$ is the same whether $Y=0$ or $Y=1$, so $H(C_i | X =0) = 0$ because register $C_i$ is then in a pure state. Similarly, on even $i$'s and for $Y=0$, the pure state messages in $C_i$ is the same whether $X=0$ or $X=1$, so  $H (C_i | B_i, Y=0) = 0$ for even $i$'s. 
Finally, for odd $i>1$'s, we notice that on $Y = 0$, the pure state held in register $B_i$ is the same as that in registers $C_{i-2} B_{i-2}$, so that the following holds:

\begin{align}
\QIC (\widetilde{\Pi}_A, \mu_0)  & = \sum_{i=1}^{2r} \QIC_i (\widetilde{\Pi}_A, \mu_0)  \\
			& \leq  \big( H(C_1 B_1 | Y = 0) - H( B_1 | Y = 0)  \\
			& \quad \quad +\sum_{i=3,~odd}^{2r-1} ( H(C_i B_i | Y = 0) - H(C_{i-2} B_{i-2} | Y = 0) \big) \\
			& =   H (C_{2r-1} B_{2r-1} | Y = 0 ) - H(B_1 | Y=0) \\
			& =   H (C_{2r-1} B_{2r-1} | Y = 0 ).
\end{align}

The content of the $C_{2r-1} B_{2r-1}$ registers at the end is $r$ copies or the same pure state, the same that was sent in message $C_1$, depending on $X$ for $Y= 0$. Denoting by $D_i$ each of these $r$ copies and using the results from Section~\ref{sec:protocol_no_ent}, we can obtain the following bound:

\begin{align}
\QIC (\widetilde{\Pi}_A, \mu_0)  & \leq  H (C_{2r-1} B_{2r-1} | Y = 0 )\\
				& =  H(D_1 \cdots D_r) \\ \label{eq:tensorH}
				& \leq h \left(\frac{1}{2} (1 - F(r, \alpha))\right),
\end{align}
with $F(r, \alpha)$ the overlap  between $r$ copies of the pure states sent in message $C_1$ corresponding to $X=0$ and $X=1$, respectively.
Recall that the overlap $F$ is $F= \exp (- \frac{1}{2} |\gamma - \delta|^2 )$ for coherent states $\ket{\gamma}$ and $\ket{\delta}$. 
We get $F(r, \alpha) = \exp (- r\frac{\abs{\alpha}^2}{2} [\sin^2 (\theta) + (1 - \cos (\theta))^2 ]) = \exp (- r\abs{\alpha}^2 [1 - \cos (\frac{\pi}{2r}) ])$. The result follows.

\section{Information Leakage Analysis for our Appointment Scheduling protocol: Experimental errors}\label{leakage_exp}

In this section we analyze the information leakage of our appointment scheduling protocol under experimental imperfections.
We also argue about the error the protocol makes.
Finally, we analyse a slight modification of the protocol in the regime $\eta = 1 - 1/r$ which allows us to achieve information leakage $\frac{n}{r}$ up to polylogarithmic terms, for $r \leq n^{1/3}$.

\subsection{Information leakage analysis}
Here we analyze the information leakage of our modified appointment scheduling protocol under experimental errors.
%

\subsubsection{$\mathrm{QIC}$ for $\Pi_D$ with experimental errors}

\begin{lemma}\label{PiD_lemma}

The protocol $\Pi_D$ run with subroutine $\Pi_A^\prime$ described in Section~\ref{sec:account_exp_error} satisfies
\begin{align*}
\QIC (\Pi_D) & \leq s + \log s + 1 + \frac{2n}{ 1 - p} p_{dark} \\
			&\quad \quad +  \frac{n}{ 1 - p} \max \Bigg[\frac{2(2r+3)}{n}, \\
&\hspace{1.5in}  \QIC_0 (\widetilde{\Pi}_A^\prime)+ 2(2r +3) \; h\left(\frac{2  \ln n }{s } + \frac{1}{n}\right)\Bigg]
\end{align*}
for all values of $n$ and $s$ satisfying $n\geq 4$ and $8\ln(n)\leq s \leq n$,
and where 
$p$ is the probability of an inconclusive outcome, i.e., the probability that either no click or a double click occur,
\begin{align*}
p&=e^{- \eta_{det}\abs{\alpha_{out}}^2}(1-p_{dark})^2+(1-e^{- \eta_{det}\abs{\alpha_{out}}^2}+e^{- \eta_{det}\abs{\alpha_{out}}^2}p_{dark})p_{dark}.
\end{align*}
\end{lemma}
In the next section we will show
\begin{align*}
\QIC_0 (\widetilde{\Pi}_A^\prime) \leq h (\frac{1}{2} (1 - \widetilde{F}(r, \alpha_{out}, \eta))),
\end{align*}
with
\begin{align*}
\tilde{F} (r, \alpha_{out}, \eta) &= \exp \left[\frac{-(\eta^{-2r} - 1)}{(1- \eta^2)} \abs{\alpha_{out}}^2 \left[1 - \cos \left(\frac{\pi}{2r}\right)\right] \right].
\end{align*}
Theorem~\ref{QIC_general_exp} follows.




First, we prove Lemma~\ref{PiD_lemma}. Define protocol $\hat{\Pi}_A$ as follows.

\begin{framed}

\textbf{Protocol} $\hat{\Pi}_A$ on inputs $ a , b \in \{0, 1 \}$:

\begin{enumerate}
\item Run Protocol $\widetilde{\Pi}_A^\prime$.
\item If $\widetilde{\Pi}_A^\prime$ returns ``0'', return output ``0''.
\item If $\widetilde{\Pi}_A^\prime$ returns ``1'', Alice and Bob exchange $a$ and $b$ and return $AND(a, b)$ as output.
\item If $\widetilde{\Pi}_A^\prime$ returns ``Inconclusive,'' return output ``Inconclusive,''.
\end{enumerate}

\end{framed}

Notice that building protocol $\Pi (\hat{\Pi}_A)$ as in Definition~\ref{def:recursive} gives the same protocol as $\Pi_A$ running subroutine $\widetilde{\Pi}_A^\prime$.
The proof of Lemma~\ref{QIC_general} applied to $\Pi(\hat{\Pi}_A)$ (and its $2r+3$ messages) yields, for $\Pi_D$ running $\widetilde{\Pi}_A^\prime$ ,
\begin{align*}
\QIC (\Pi_D) & \leq s + \log s + 1 + \frac{n}{ 1 - p} \max \Bigg[\frac{2(2r+3)}{n}, \\
&\hspace{1.5in} \QIC_0 (\hat{\Pi}_A) + 2(2r +3)  \; h \left(\frac{2  \ln n }{s } + \frac{1}{n}\right)\Bigg]
\end{align*}
for all values of $n$ and $s$ satisfying $n\geq 4$ and $8\ln(n)\leq s \leq n$.
The  result follows by noting that 
\begin{align*}
\QIC_0 (\hat{\Pi}_A) \leq \QIC_0 (\widetilde{\Pi}_A^\prime) + 2p_{dark},
\end{align*} 
and that $2p_{dark}$ is non-negative so we can take it out of the maximization to simplify the expression.
The inequality holds since the only extra information leaked by $\hat{\Pi}_A$ after running $\widetilde{\Pi}_A^\prime$ is the exchange of $a$ and $b$. Moreover,  on distribution $\mu_0$ as considered in $\QIC_0$, this can only occur if there has been a dark count on the detector corresponding to the ``1'' output, which occurs with probability at most $p_{dark}$. Using that conditioning on a common variable is taking average (Lemma~\ref{lem:qic_cond_av}) along with a dimension bound (Lemma~\ref{prelim:lem:CQMIfacts}) then limits the extra information leakage to two bits scaled by probability $p_{dark}$.

\subsubsection{Definition of $\QIC$ for lossy $AND$ protocol}
In this section we detail the framework we will use to calculate $\QIC_0(\widetilde{\Pi}_A^\prime)$. A lossy channel 
can be modeled by a beamsplitter with transmissivity $\eta$. We assume that channel loss resides in the communication register during transmission, and after transmission it resides in the receiving party's memory but he does not access it.

Recall the general expression for the information leakage
\begin{align}
\QIC(\widetilde{\Pi}_A^\prime,\mu_0)=\sum_{i=1}^{2r+1} I(X; C_i | Y B_i) + I(Y; C_i | X A_i).
\end{align}

{First we simplify this expression.} 
As with the ideal protocol, 
we first handle the last message.
Under $\mu_0$, this last message is $0$ with probability $1- e^{-\eta_{det}\abs{\alpha_{out}}^2}$ and inconclusive with probability $e^{-\eta_{det}\abs{\alpha_{out}}^2}$. This holds independently of the input $(x, y)$ and of the content of  $A_{r+1}$ and $B_{2r+1}$. Thus, $\QIC_{2r+1} (\widetilde{\Pi}_A, \mu_0) = 0$.
This is the only non-pure message in the protocol.

The first $2r$ messages are pure, so we start from (\ref{eq:purestateinfo}).
We first note that under $\mu_0$, if $X=1$ then $Y=0$, so $H(C_i | A_i, X =1) = 0$ because register $C_i$ is then in a pure state (Lemma~\ref{prelim:lem:CQMIfacts}). Similarly, if $Y=1$ then $X=0$ and $H(C_i | B_i, Y =1) = 0$ because register $C_i$ is then in a pure state.
Further,  on $X=0$, the pure state messages in $C_i$ is the same whether $Y=0$ or $Y=1$, so $H(C_i | A_i, X =0) = 0$ because register $C_i$ is then in a pure state. Similarly, on even $i$'s and for $Y=0$, the pure state messages in $C_i$ is the same whether $X=0$ or $X=1$, so  $H (C_i | B_i, Y=0) = 0$ for even $i$'s. 
This gives

\begin{align}\label{exp_qic_AND}
\QIC(\widetilde{\Pi}_A^\prime,\mu_0)\leq \sum_{i=1, \text{ odd}}^{2r} H(C_i B_i | Y=0)-H(B_i | Y=0).
\end{align}
{In the following section we apply the above bound to our subroutine $AND$ protocol $\widetilde{\Pi}_A^\prime$.}

\subsubsection{Analysis of $\QIC$ for lossy $AND$ protocol}

Here we prove the following bound on the information leakage $\QIC_0(\widetilde{\Pi}_A^\prime) = \max_{\mu_0: \mu_0 (1,1 ) = 0} \QIC(\widetilde{\Pi}_A^\prime, \mu_0) $ of the subroutine $\widetilde{\Pi}_A^\prime$ to the $AND$ protocol ${\Pi}_A$ (and $\hat{\Pi}_A$).
\begin{lemma}
\begin{align*}
\QIC_0 (\widetilde{\Pi}_A^\prime) \leq h(\frac{1}{2} (1 -\prod_{i=1}^r F_i)),
\end{align*}
in which
\begin{align}\label{pi2expfid}
F_i= \exp\left[-\frac{\abs{\alpha_i}^2}{\eta}\left[1-\cos \left(\frac{\pi}{2r}\right)\right] \right]
\end{align}
for all $i=1, 2, \dots, r$.
\end{lemma}

%
%

The result stated in the preceding section follows by simplying the product of $F_i$'s, the corresponding sum of $|\alpha_i|^2$ being a geometic series.

From expression \eqref{exp_qic_AND} it is clear that any content of $B_i$ which produces an uncorrellated pure state when conditioned on $Y=0$ can be safely discarded without changing the information cost. Therefore, we assume $B_i$ contains only elements which do not produce an uncorellated pure state when conditioned on $Y=0$.
Under this assumption, the state of the registers $C_i B_i$ for odd $i$ in the $(0,0), (0,1), (1,0)$ cases are as follows:\\
\begin{samepage}
\textbf{State of registers $C_i B_i$ for odd $i$ on different inputs for protocol $\widetilde{\Pi}_A$:}

\textbf{On (0,0):}\\
\indent $i$ odd $(A \rightarrow B)$:
\begin{align}
\left(\ket{\frac{1}{\eta^{r-i/2}}\alpha,0}\ket{\frac{\sqrt{1-\eta}}{\eta^{r-(i-1)/2}}\alpha,0}\right)
\bigotimes_{l=1,\mathrm{odd}}^{i-2}\left(\ket{\frac{1}{\eta^{r-l/2}}\alpha,0}\ket{\frac{\sqrt{1-\eta}}{\eta^{r-(l-1)/2}}\alpha,0}\right)\nonumber
\end{align}
\textbf{On (0,1):} Identical to (0,0).\\
\textbf{On (1,0):}\\
\indent $i$ odd $(A \rightarrow B)$:
\begin{align}
&\left(\ket{\frac{1}{\eta^{r-i/2}}\alpha\cos 2\theta,\frac{1}{\eta^{r-i/2}}\alpha\sin 2\theta}\ket{\frac{\sqrt{1-\eta}}{\eta^{r-(i-1)/2}}\alpha\cos2\theta,\frac{\sqrt{1-\eta}}{\eta^{r-(i-1)/2}}\alpha\sin2\theta}\right)
\nonumber\\
&\bigotimes_{k=1,\mathrm{odd}}^{i-2}\left(\ket{\frac{1}{\eta^{r-k/2}}\alpha\cos 2\theta,\frac{1}{\eta^{r-k/2}}\alpha\sin 2\theta}\ket{\frac{\sqrt{1-\eta}}{\eta^{r-(k-1)/2}}\alpha\cos2\theta,\frac{\sqrt{1-\eta}}{\eta^{r-(k-1)/2}}\alpha\sin2\theta}\right)\nonumber
\end{align}
\end{samepage}
Where the first two modes are contained in register $C_i$ and the rest are contained in register $B_i$. As before, the content of registers $C_{i-2} B_{i-2}$ is identical to that of register $B_i$. By nearly identical arguments to that of Appendix~\ref{QIC_subroutine},
\begin{align}
\QIC_0(\widetilde{\Pi}_A(\eta))  &\leq  h (C_{2r-1} B_{2r-1} | Y = 0 )\nonumber\\
						&\leq h \left(\frac{1}{2} (1 - \prod_{i=1}^r F_i)\right),
\end{align}
where $\prod_{i=1}^r F_i$ is the overlap of the two possible states of $C_{2r-1} B_{2r-1}$ when $Y=0$ and 
is given by Eqn.~\eqref{pi2expfid}.

\subsection{Error of $\Pi_D$ with experimental errors}

We analyse the error in $\Pi_D$ run with experimental errors. First, whenever an index $i$ is output, there is never any error since such an output can only come after a classical verification that $x_i = y_i =1$. Hence, the only potential error arise if $\Pi_D$ outputs $\emptyset$, and we show that the protocol errs with probability at most $p_{dark}$ on that output. 

Assume that  there is one or more intersection, so that $\Pi_D$ should not output $\emptyset$. What is the probability to still have that output? In particular, considering the first such intersection, there must have been a dark count when $\Pi_A$ was run for that date, together with other independent events. Hence, the probability to output $\emptyset$ is at most $p_{dark}$.

\subsection{Information leakage in the $\eta = 1 - 1/r$ regime}\label{sec:modif-one-over-r}

We now prove an upper bound of $n/r$ up to log terms on quantum information leakage for number of round $r$ in $\widetilde{\Pi}_A^\prime$ satisfying $r \leq n^{1/3}$ as $n$ grows, and with $\eta = 1-1/r$. First, note that $\eta^r \rightarrow e^{-1}$ as $r$ increases, so that $h (1/2 (1-F(r, \alpha_{out}, \eta)))$ scales as $1/r$ up to log terms, as in the ideal case. We now handle the $\frac{2n}{1-p} p_{dark}$ term. For this, we modify $\Pi_A$ as follows.

\begin{framed}

Protocol $\Pi_A^{\prime \prime \prime}$ on inputs $a, b \in \{ 0, 1 \}$:

\begin{enumerate}
\item Run Protocol $\widetilde{\Pi}_A$ $O(\log n)$ times.
\item If get majority of ``$1$'', exchange $a$, $b$ and output $AND(a, b)$.
\item Else, output ``$0$''.
\end{enumerate}

\end{framed}

Choose the constant in $O(\log n)$ according to experimental parameters and such that the probability to get output ``$1$'' when $AND(a,b) = 0$ is less than $1/n^2$. It follows from subadditivity (Lemma~\ref{lem:prop_QIC}), from the dimension bound on two-bit inputs (Lemma~\ref{prelim:lem:CQMIfacts}) and the fact that conditioning on common variable is taking average (Lemma~\ref{lem:qic_cond_av}) that

\begin{align}
QIC_0 (\Pi_A^{\prime \prime \prime}) \leq O(\log n) \; QIC_0 (\widetilde{\Pi}_A^\prime) + 2/n^2 .
\end{align}

Then $QIC(\Pi_D) \leq n/r$ up to logarithmic terms. This follows from an analysis similar to Lemma~\ref{QIC_general} (avoiding the “rerun an inconclusive protocol” part, since here $\Pi_A^{\prime \prime \prime}$ is explicitly run $O(\log n)$ times).

\section{Review of coherent state mapping}\label{mappingrev}

Arrazola and L\"utkenhaus recently proposed a mapping from any quantum protocol which uses pure quantum states, unitary operations, and projective measurements to a corresponding coherent state protocol \cite{PhysRevA.95.032337}. In this section we review the general mapping, and in the next sections we apply it to two existing protocols to develop new coherent state appointment scheduling protocols.

The general coherent state mapping proceeds as follows. Define a function {$f_{\alpha}: \mathds{C}^n \rightarrow L^2(\mathds{R})^{\otimes n}$} as
\begin{align}\label{cohfunc}
f_{\alpha}\left( \sum_{i=1}^n \lambda_i \ket{i}\right) = \bigotimes_{i=1}^n \ket{\lambda_i \alpha}_i.
\end{align}
Given a vector $\ket{\psi}\in \mathds{C}^n$, we use the shorthand $\ket{\psi_\alpha}:=f_{\alpha}(\ket{\psi})$. From \eqref{cohfunc} it follows that $\ket{\psi_\alpha}$ will have total mean photon number $\abs{\alpha}^2$ for all unit vectors $\ket{\psi}\in \mathds{C}^n$. For any unitary $U\in \U(\mathds{C}^n)$, the unitary {$V\in \U(L^2(\mathds{R}))^{\otimes n})$} which transforms the modes as
{
\begin{align}\label{unmode}
a_j^\dg \rightarrow \sum_{i=1}^n U_{j, i} a_i^{\dg}
\end{align}
}
can be shown to satisfy $V\ket{\psi_\alpha}=f_{\alpha}(U\ket{\psi})$ for all $\ket{\psi}\in \mathds{C}^n$. Thus, $V$ does not change the total mean photon number $\abs{\alpha}^2$.


We note that Arrazola and L{\"u}tkenhaus also showed that if single photon detection is performed on each mode of $\ket{\psi_\alpha}$ then the probability distribution of the number of photons measured in each mode is equal to that obtained from repeated canonical basis measurements of $\ket{\psi}$, where the number of repetitions is drawn from a Poisson distribution with mean $\abs{\alpha}^2$.
Furthermore, Arrazola and L\"utkenhaus showed that due to the fact that the states always have total mean photon number $\abs{\alpha}^2$, they mostly reside in a O(log n) qubit ``typical subspace'' which includes the span of all Fock states with total photon number lying in a neighbourhood of $\abs{\alpha}^2$. In Appendix~\ref{qicohg} we adapt this result, and show that the fixed total mean photon number also implies that the mapping has low information leakage in the interactive communication setting.

\section{A second coherent state bitwise-$AND$ appointment scheduling protocol}\label{app2}
In this section we develop a coherent state appointment scheduling protocol which uses a bitwise-$AND$ subroutine protocol $\widetilde{\Pi}_A^{\prime \prime}$ that is the coherent state mapping of a protocol developed by the authors of~\cite{JainRS03}.  We then analyze its information leakage, and find that it is greater than $\widetilde{\Pi}_A$.

\subsection{Description of Protocol}

The only difference between our bitwise-$AND$ protocol developed in the main text and the protocol we describe here is in the subroutine $\widetilde{\Pi}_A$. Every other step in the protocol is identical to that of our original bitwise-$AND$ protocol, in both the ideal and experimental settings. We focus here only on the ideal setting. To describe our protocol $\widetilde{\Pi}_A^{\prime \prime}$, we first review the qubit $AND$ protocol due to Jain, Radhakrishnan and Sen (this protocol was recently described in \cite{braverman2015near, chailloux2017information}, and is reviewed below).

On inputs $x,y\in \{0,1\}$ given to Alice and Bob respectively, the following protocol computes $AND(x,y)$ in $r$ rounds for any even positive integer $r$.

First, let $\theta=\frac{\pi}{4r}$ and $\ket{v}=\cos(\theta) \ket{0}+\sin(\theta) \ket{1}$. Let $U_v$ be the unitary operator reflecting about the vector $\ket{v}$, i.e. $U_v \ket{0}=\cos(2\theta)\ket{0}+\sin(2\theta) \ket{1}$ and $U_v \ket{1}=\sin(2\theta)\ket{0}-\cos(2\theta) \ket{1}$. Let $U_0$ be the operator reflecting about $\ket{0}$, i.e. $U_0 \ket{0}=\ket{0}$ and $U_0 \ket{1}=-\ket{1}$.

The unambiguous qubit $AND$ protocol of the authors of \cite{JainRS03} proceeds as follows.

\begin{samepage}
\begin{framed}
\textbf{Qubit $AND$ protocol of Jain, Radhakrishnan and Sen}

First, Alice prepares a qubit-register $C$ initialized to the state $\ket{0}$. Then, on each round, Alice and Bob do the following:
\begin{enumerate}
\item On $x=0$ ($x=1$), Alice performs the identity map ($U_v$ map) on the register $C$ and sends it to Bob.
\item On $y=0$ ($y=1$), Bob performs the identity map ($U_0$ map) on the register $C$ and sends it to Alice.
\end{enumerate}
After $r$ rounds the state of register $C$ will be $\ket{0}$ $\left(-\ket{1}\right)$ if $AND(x,y)=0$ $(1)$. Alice measures $C$ in the standard basis to determine the result, which she communicates to Bob. Clearly this is an unambiguous two-bit $AND$ protocol with zero probability of an inconclusive outcome.
\end{framed}
\end{samepage}

To construct $\widetilde{\Pi}_A^{\prime \prime}$, we apply the coherent state mapping described in Appendix~\ref{mappingrev} to the above qubit protocol. In contrast to the qubit protocol, $\widetilde{\Pi}_A^{\prime \prime}$ has some probability $p$ of an inconclusive outcome, just like the protocol $\widetilde{\Pi}_A$.

Define a unitary $V_0$ as
\begin{align}
V_0 \ket{\alpha}\ket{\beta}=\ket{\alpha}\ket{-\beta},
\end{align}
which acts as a phase flip on the second mode. Clearly, $V_0 f_\alpha (\ket{\psi})=f_\alpha(U_0 \ket{\psi})$ for every state $\ket{\psi}$ used in the qubit protocol.

Define a unitary $R_\theta$ as
\begin{align}
R_\theta \ket{\alpha}\ket{\beta}=\ket{\cos(\theta) \alpha-\sin(\theta)\beta}\ket{\sin(\theta) \alpha+\cos(\theta)\beta}
\end{align}
which acts as a beamsplitter specified by angle $\theta$. Define a unitary {$V_v=R_\theta V_0 R_\theta^\dg$}, where ${R_\theta^\dg=R_{-\theta}}$. It can be shown that $V_v f_\alpha (\ket{\psi})=f_\alpha(U_v \ket{\psi})$ for every state $\ket{\psi}$ used in the qubit protocol.

The protocol $\widetilde{\Pi}_A^{\prime \prime}$ then proceeds as follows.
\begin{samepage}
\begin{framed}
\textbf{Coherent state mapping $\widetilde{\Pi}_A^{\prime \prime}$ of qubit $AND$ protocol}

First, Alice prepares a two-mode register $C$ in state $\ket{\alpha,0}$, for some $\alpha > 0$. On each of the $r$ rounds, Alice and Bob do the following:
\begin{enumerate}
\item On $x=0$ ($x=1$), Alice performs the identity map ($V_v$ map) on the register $C$ and sends it to Bob.
\item On $y=0$ ($y=1$), Bob performs the identity map ($V_0$ map) on the register $C$ and sends it to Alice.
\end{enumerate}
After $r$ rounds, Alice measures each mode of $C$ with single photon threshold detectors and communicates the result to Bob.
\end{framed}
\end{samepage}

In ideal implementations, after all unitaries are performed Alice ends up with $\ket{\alpha, 0}$ on inputs $(0, 0)$, $(0, 1)$ and $(1, 0)$, and with $\ket{0, - \alpha}$ on input $(1, 1)$.  Thus, she might detect a photon in the first mode only if the output to $AND$ is $0$ and she might detect a photon in the second mode only if the output to $AND$  is $1$. If she does not detect any photon, she tells Bob that the run was inconclusive. Note that Alice obtains a click with probability $1 - e^{-\abs{\alpha}^2}$ for any input. Thus, this protocol never outputs a wrong answer, and has some uniform probability $p = e^{- \abs{\alpha}^2}$ of outcome ``Inconclusive''. For clarity, we explicitly write down how the protocol evolves for different inputs:

\begin{samepage}
\textbf{Evolution of $\widetilde{\Pi}_A^{\prime \prime}$ for different inputs:}

\textbf{On (0, 0):}
\begin{tabular}{l l l l}
$\ket{\alpha, 0}$ & $\rightarrow_A \ket{\alpha, 0}$ & $\rightarrow_B \ket{\alpha, 0}$ & $\rightarrow_A \cdots $ 
\end{tabular}

\textbf{On (0, 1):}
\begin{tabular}{l l l l}
$\ket{\alpha, 0}$ & $\rightarrow_A \ket{\alpha, 0}$ & $\rightarrow_B  \ket{\alpha, 0}$ & $\rightarrow_A \cdots $ 
\end{tabular}

\textbf{On (1, 0):} 
\begin{tabular}{l l l}
		$\ket{\alpha, 0}$ & $\rightarrow_A \ket{\cos (2 \theta) \alpha, \sin (2 \theta ) \alpha}$
		&$\rightarrow_B  \ket{\cos (2 \theta) \alpha, \sin (2 \theta ) \alpha}$ \\
		 &$\rightarrow_A \ket{\alpha, 0} $
		 &$\rightarrow_B \ket{\alpha, 0} $ \\
& $\rightarrow_A \ket{\cos (2 \theta) \alpha, \sin (2 \theta ) \alpha}$
		&$\rightarrow_B  \ket{\cos (2 \theta) \alpha, \sin (2 \theta ) \alpha}$ \\
		 &$\rightarrow_A \ket{\alpha, 0} $
		 &$\rightarrow_B \ket{\alpha, 0} $ \\
& \vdots & \vdots
\end{tabular}

\textbf{On (1, 1):}
\begin{tabular}{l l l}
$\ket{\alpha, 0}$ & $\rightarrow_A \ket{\cos (2 \theta) \alpha, \sin (2 \theta ) \alpha}$ &\hspace{-.3in} $\rightarrow_B \ket{\cos (2 \theta) \alpha, - \sin (2\theta ) \alpha}$ \\
& $\rightarrow_A \ket{\cos (4 \theta) \alpha, \sin (4 \theta ) \alpha}$ & $\hspace{-.25in}\rightarrow_B \ket{\cos (4 \theta) \alpha, - \sin (4 \theta ) \alpha}$ \\
& \vdots & \vdots \\
& $\rightarrow_A \ket{\cos (2 r \theta) \alpha, \sin (2 r \theta ) \alpha}$ & $\hspace{-.15in}\rightarrow_B \ket{\cos (\frac{\pi}{2}) \alpha , - \sin (\frac{\pi}{2}) \alpha}$\\
& =$\ket{0,-\alpha} $
\end{tabular}
\end{samepage}
On (0,0) and (0,1) Alice and Bob's manipulations leave the state unchanged. On (1,0) Alice performs $V_v$ and Bob does nothing. Since $V_v$ is its own inverse, the state oscillates between two forms in this case. On (1,1) Alice and Bob's manipulations bring the state to $\ket{0,-\alpha}$ after $r$ rounds.

Just like in protocol $\widetilde{\Pi}_A$, after $r$ rounds Alice measures each mode of register $C$ with single photon threshold detectors and 
communicates the result to Bob.

As noted previously, the remainder of this new appointment scheduling protocol proceeds identically to our previous appointment scheduling protocol with $\widetilde{\Pi}_A$ replaced by $\widetilde{\Pi}_A^{\prime \prime}$.

\subsection{Information leakage analysis}

Here we analyze the information leakage of protocol $\widetilde{\Pi}_A^{\prime \prime}$. 
By the previous analysis, we need only bound $\QIC_0(\widetilde{\Pi}_A^{\prime \prime})$, as this implies a bound on the information leakage of the full appointment scheduling protocol. 
We prove the following:

\begin{lemma}
\begin{align*}
\QIC_0(\widetilde{\Pi}_A^{\prime \prime}) \leq r \left[h\left(\frac{1}{2} (1 - F'(r, \alpha))\right)\right],
\end{align*}
in which
\begin{align}\label{pi1expfid}
F'(r, \alpha)= \exp\left[-\abs{\alpha}^2\left[1-\cos \left(\frac{\pi}{2r}\right)\right] \right].
\end{align}
\end{lemma}

This result implies the following bound on $\QIC \left(\Pi_D \right)$.
\begin{corollary}

The protocol $\Pi_D$ constructed from protocol $\widetilde{\Pi}_A^{\prime \prime}$ satisfies
\begin{align*}
\QIC \left(\Pi_D \right) & \leq s + \log s + 1 \\
	&\quad \quad + \frac{n}{ 1 - \exp(-|\alpha|^2)} \max \Bigg[\frac{2(2r+1)}{n}, \nonumber\\
 & \quad \quad \quad \quad r \; h\left(\frac{1}{2} (1 - F'(r, \alpha))\right)
+2(2r+1) \; h\left(\frac{2  \ln n }{s } + \frac{1}{n}\right)\Bigg]
\end{align*}
in which
\begin{align}\label{pi1expfid}
F'(r, \alpha)= \exp\left[-\abs{\alpha}^2\left[1-\cos \left(\frac{\pi}{2r}\right)\right] \right].
\end{align}
\end{corollary}

Note that, by the chain rule and the data processing inequality, 
\begin{align}\label{entropyinequality}
\sum_{i=1}^r \left[h\left(\frac{1}{2} (1 - F_i)\right)\right] \geq h\left(\frac{1}{2} (1 - \prod_{i=1}^r F_i)\right)
\end{align}
for any $F_i \leq 1$, $i =1,\dots, r$, 
so this protocol has greater information leakage than our protocol developed in the main text, under the bounds we have used. 

Extending the analysis to the setting with experimental errors, we have found numerically that the corresponding protocol $\Pi_D$ running subroutine $\Pi_A^{\prime \prime}$ (a variant made robust to experimental errors) beats the classical lower bound by a factor of two under experimental imperfections $\eta=0.995$, $p_{dark}=4\times 10^{-8}$, and $\eta_{det}=0.9$. Note that the required transmissivity is much higher than the transmissivity $\eta=0.99$ needed to obtain a factor of two improvement in our protocol $\Pi_D$ running $\Pi_A^\prime$ in the main text.

Now we prove the lemma. 

We first handle the last message.
Under $\mu_0$, this last message is $0$ with probability $1- e^{-\abs{\alpha}^2}$ and inconclusive with probability $e^{-\abs{\alpha}^2}$. This holds independently of the input $(x, y)$. Thus, $\QIC_{2r+1} (\widetilde{\Pi}_A, \mu_0) = 0$.
This is the only non-pure message in the protocol.

The first $2r$ messages are pure, so we start from (\ref{eq:purestateinfo}), for which we have  trivial $A_i$ and $B_i$ registers here.
We first note that under $\mu_0$, if $X=1$ then $Y=0$, so $H(C_i | X =1) = 0$ because register $C_i$ is then in a pure state (Lemma~\ref{prelim:lem:CQMIfacts}). Similarly, if $Y=1$ then $X=0$ and $H(C_i | Y =1) = 0$ because register $C_i$ is then in a pure state.
Further,  on $X=0$, the pure state messages in $C_i$ is the same whether $Y=0$ or $Y=1$, so $H(C_i | X =0) = 0$ because register $C_i$ is then in a pure state. Similarly, in even rounds and for both Alice's and Bob's messages,  for $Y=0$, the pure state messages in $C_i$ is the same whether $X=0$ or $X=1$, so  $H (C_i | Y=0) = 0$ the corresponding $i$'s. 

The only non-zero terms thus correspond to odd rounds and $Y=0$.
In odd rounds, on $Y=0$, Alice's and Bob's messages are always the same. (In fact, this is the same message that appeared in the information cost under $\mu_0$ for our main protocol for $AND$.) Thus, the non-zero terms are always the same, as in $C_1$, and there are $r$ of them: one for Alice and one for Bob in each of the $r/2$ odd rounds. We get
\begin{align*}
\QIC (\widetilde{\Pi}_A^{\prime \prime }, \mu_0)  & = \sum_{i=1}^{2r} \QIC_i (\widetilde{\Pi}_A^{\prime \prime}, \mu_0)  \\
			& \leq  r \; H(C_1 | Y = 0)   \\
	&\leq r \; h\left(\frac{1}{2} (1 - F'(r, \alpha))\right),
\end{align*}
where $F'(r, \alpha)$ is the fidelity between the two possible states of register $C_{1}$ when $Y=0$, and is given by \eqref{pi1expfid}. This completes the proof.

\section{Coherent state version of distributed Grover search protocol}\label{cohgrov}
In this section we describe the distributed Grover search protocol of \cite{Buhrman:1998:QVC:276698.276713}, based on~\cite{boyer1996tight}, and then proceed to describe our implementation of the protocol's coherent state mapping (defined in Appendix~\ref{mappingrev}), for which we find a practical implementation of the oracle calls. We then proceed to show that, in the ideal setting,  the information leakage of this protocol is $\mathcal{O}(\sqrt{n} \log n)$, just like the original distributed Grover search protocol.

\subsection{Original distributed Grover search protocol}

In the distributed Grover search protocol of \cite{Buhrman:1998:QVC:276698.276713} Alice and Bob receive $x, y \in\{0,1\}^n$. We assume for now that they either have no intersection or intersect in $k$ unknown indices $a_1,\dots a_k$. These works describe a protocol which uses $\mathcal{O}(\sqrt{n} \log n)$ qubits of communication to either find a common intersection or determine with high probability that $x$ and $y$ do not intersect. In this section we review the protocol of \cite{Buhrman:1998:QVC:276698.276713} under the simplifying assumption $k\ll n$. (In order to ensure that this assumption is satisfied with high probability, Alice and Bob could perform a classical subsampling of, say, $\sqrt{n}$ dates as in protocol $\Pi_D$ of the main text before running the protocol we describe here.) We first consider the case in which $k$ is known, and briefly discuss the extension to unknown $k$ at the end of this section.

Now we describe the distributed Grover search protocol, which always outputs $``\emptyset"$ on non-intersecting inputs, and outputs $``\emptyset"$ on intersecting inputs with probability at most $\epsilon$ (with probability $1-\epsilon$ it finds an intersection).
\begin{samepage}
\begin{framed}
\textbf{Distributed Grover search appointment scheduling protocol}\\
First, Alice prepares the state
\begin{align}
\ket{s}=\frac{1}{\sqrt{n}}\sum_{i=1}^n \ket{i}.
\end{align}
Choose iteration number $r=\fl{\pi/(4\theta)}$, for $\theta$ satisfying $\sin^2 \theta=k/n$. Then the following is iterated $r$ times:
\begin{enumerate}
\item Alice and Bob jointly perform the oracle call unitary
\begin{align}
U_A=\mathds{1}-2\sum_{j=1}^k\kb{a_j}{a_j}
\end{align}
using the protocol outlined below.
\item Alice performs the inversion about the mean unitary
\begin{align}
U_S&=2\kb{s}{s}-\mathds{1}.
\end{align}
\end{enumerate}
Then, Alice measures the state in the canonical basis, obtaining some outcome $i \in [n]$, and sends $(i, x_i)$ to Bob. Bob then sends $y_i$ to Alice. If they find that $x_i=y_i=1$, they output this index. Otherwise, they repeat the protocol. If they repeat the protocol $K=\cl{\log(1/\epsilon)/\log(n/k)}$ times without finding an intersection, they output $``\emptyset"$.
\end{framed}
\end{samepage}
The iteration number $r$ is chosen as above because if  $x$ and $y$ intersect in $k$ indices, then the probability that Alice's measurement produces a non-intersecting index $i$ is no greater than $k/n$ under this choice (as discussed further below). The repetition number $K$ is chosen to attain error probability $\epsilon$.

In more details, the protocol evolves as follows. Let
\begin{align}
\ket{t}=\frac{1}{\sqrt{1-(k/n)}}\left( \ket{s}-\frac{1}{\sqrt{n}}\sum_{j=1}^k \ket{a_i}\right)
\end{align}
and
\begin{align}
\ket{\tilde{a}}=\frac{1}{\sqrt{k}}\sum_{j=1}^k \ket{a_j},
\end{align}
then as shown in \cite{boyer1996tight}, after $l$ applications of $U_S U_A$ the state is given by
\begin{align}\label{grovelution}
(U_S U_A)^l \ket{s}&=\sin((2l+1)\theta) \ket{\tilde{a}}+\cos((2l+1)\theta) \ket{t},
\end{align}
for $\theta$ defined as above. In \cite{boyer1996tight} it is shown that for $r=\fl{\pi/(4\theta)}$, the probability $\cos^2((2r+1)\theta)$ that Alice's measurement does not output an intersecting index $i$ satisfies $\cos^2((2r+1)\theta) \leq k/n$.

Now we detail how Alice and Bob jointly perform the oracle call unitary $U_A$.
Let $U_x, U_y \in \U(\mathds{C}^n \otimes \mathds{C}^2)$ act as
\begin{align}
U_x \ket{i}\ket{z}&=\ket{i}\ket{x_i \oplus z} \text{for all $i=1, \dots, n$}\\
U_y \ket{i}\ket{z}&=\ket{i}\ket{y_i \oplus z} \text{for all $i=1, \dots, n$},
\end{align}
(which Alice and Bob can implement, respectively), $W\in \U(\mathds{C}^2 \otimes \mathds{C}^2)$ is the swap operator which acts as
\begin{align}
W\ket{i}\ket{j} =\ket{j}\ket{i} \text{for all $i,j=1,2$},
\end{align}
and $V$ is the control-$U_y$ gate, where $U_y$ acts on the first two systems, and the state of the third system is the control.

\begin{samepage}
\begin{framed}
\textbf{Procedure to implement oracle call unitary $U_A$}

Alice prepares auxilliary qubits $\ket{0}\ket{-}$, so the state of her entire register is $\ket{\psi}\ket{0}\ket{-}$, where $\ket{\psi}\in \mathds{C}^n$ is the resultant state from the previous step in the appointment scheduling protocol.
\begin{enumerate}
\item  Alice applies $(U_x \otimes \mathds{1}_2)$ and sends the entire state to Bob.
\item  Bob applies $(\mathds{1}_n\otimes W)(V)(\mathds{1}_n \otimes W)$ and sends the entire state back to Alice.
\item Alice applies $(U_x \otimes \mathds{1}_2)$, and discards the qubits $\ket{0}\ket{-}$.
\end{enumerate}
\end{framed}
\end{samepage}

It is straightforward to show that
\begin{align}
(U_x \otimes \mathds{1}_2)(\mathds{1}_n\otimes W)(V)(\mathds{1}_n \otimes W)(U_x \otimes \mathds{1}_2) \ket{i} \ket{0} \ket{-}= (U_A \otimes \mathds{1}_2 \otimes \mathds{1}_2) \ket{i} \ket{0}\ket{-}
\end{align}
for all $i \in [n]$, so the above procedure implements $U_A$.

Thus, for each application of $U_A$, Alice and Bob exchange $2(\log(n) +2)$ qubits. For $k \ll n$, $U_A$ must be implemented at most $Kr=\mathcal{O}(\sqrt{n/k})$ times. For each repetition of the protocol, Alice sends Bob her measurement outcome $i$ (which is $\log n$ bits) along with $x_i$ (which is one bit), and Bob sends Alice $y_i$ (which is one bit). Thus, the amount of communication in these stages is upper bounded by $K(\log n+2)=\mathcal{O}(\log n)$ bits. Thus, the protocol uses a total of $\mathcal{O}(\sqrt{n/k}\log(n))$ qubits of communication. By the dimension bound, the information leakage of this protocol is also $\mathcal{O}(\sqrt{n/k}\log(n))$.

Now we consider the case in which $k$ is unknown to either party, but is known to be much less than $n$. The implementation of the unitaries $U_S$ and $U_A$ is independent of $k$, so they can still be applied, but the iteration number $r$ is a function of $k$ (and $n$), and must now be chosen in a different manner. The protocol proposed in \cite{boyer1996tight} uses a randomized algorithm to choose the iteration number $r$, and finds a common intersection (or determines no intersection with high probability) while maintaining the $\mathcal{O}(\sqrt{n}\log(n))$ behaviour.

\subsection{Coherent state distributed Grover search with practical oracle calls}

We proceed to describe the coherent-state mapping of Appendix~\ref{mappingrev} applied to the distributed Grover search protocol. We find a protocol for which the linear optics transformation $V_A$ corresponding to the oracle call $U_A$ uses only local phase shifters and the swapping of two modes. Unfortunately, the linear optics transformation $V_S$ corresponding to the inversion about the mean $U_S$ still requires a global transformation of the state. In Appendix~\ref{qicohg} we prove that this protocol has information leakage $\mathcal{O}(\sqrt{n} \log n)$, a nearly quadratic improvement over the classical information leakage lower bound of $\Omega({n})$ proven in \cite{braverman2013information} and \cite{1611.06650} for the zero-error and nonzero-error cases, respectively.

We again consider only the case in which $x$ and $y$ either have no intersection or intersect in $k$ unknown indices $a_1,\dots a_k$ for $k\ll n$. We suggest that this protocol could be adapted in similar fashion to \cite{boyer1996tight} if this is not the case.

We first describe the coherent state mapping of the distributed Grover search protocol in terms of $V_S$ and $V_A$. Let $\ket{\psi}=\sum_{i=1}^n \lambda_i \ket{i} \in \mathds{C}^n$ be an arbitrary pure state, which will help us describe the action of $V_S$ and $V_A$. The following coherent state mapping of the distributed Grover search protocol always outputs $``\emptyset"$ on non-intersecting inputs, and outputs $``\emptyset"$ on intersecting inputs with probability at most $\epsilon$ (with probability $1-\epsilon$ it finds an intersection).

\begin{samepage}
\begin{framed}
\textbf{Coherent state mapping of distributed Grover search protocol}\\
For some constant $\alpha \in \mathds{C}$ (which can be optimized over), Alice prepares the state
\begin{align}
\bigotimes_{i=1}^n \ket{\alpha/\sqrt{n}}_i
\end{align}
Choose iteration number $r=\fl{\pi/(4\theta)}$, for $\theta$ satisfying $\sin^2 \theta=k/n$. Then the following is repeated $r$ times:
\begin{enumerate}
\item Alice and Bob jointly perform the linear optics transformation $V_A$ corresponding to the oracle call $U_A$, which acts as
\begin{align}\label{VA}
V_A f_{\alpha}(\ket{\psi})=\bigotimes_{i=1}^n \ket{(-1)^{x_i \wedge y_i} \lambda_i \alpha}_i,
\end{align}
using the protocol outlined below.
\item Alice performs the linear optics transformation $V_S$ corresponding to the inversion about the mean $U_S$, which acts as
\begin{align}
V_S f_{\alpha}(\ket{\psi})=\bigotimes_{i=1}^n \ket{(2 v- \lambda_i)\alpha}_i,
\end{align}
for $\nu= (\lambda_1 + \dots + \lambda_n)/n$.
\end{enumerate}
Alice measures each mode with single photon threshold detectors. If no detectors click, she announces this and the parties repeat the protocol. Otherwise, she chooses a random index $i$ for which she received a click, and sends $(i, x_i)$ to Bob. Bob then sends $y_i$ to Alice. If $x_i=y_i=1$, the parties output this index. Otherwise, they repeat the protocol. If they repeat the protocol
\begin{align}
K=\cl{{\log(1/\epsilon)}\Bigg/{\log\left(\frac{1}{1-e^{-\abs{\alpha}^2\frac{k}{n}}+e^{-\abs{\alpha}^2}}\right)}}
\end{align}
times without finding an intersection, they output $``\emptyset"$.
\end{framed}
\end{samepage}


The iteration number $r$ is chosen as above because if  $x$ and $y$ intersect in $k$ indices, then the probability that Alice's measurement produces a non-intersecting index $i$ is no greater than $(1-e^{-\abs{\alpha}^2\frac{k}{n}})$ under this choice (as discussed further below). The repetition number $K$ is chosen to attain error probability $\epsilon$. The extra term $e^{-\abs{\alpha}^2}$ is the probability that no clicks occur.

In more details, the protocol evolves as follows. After $l$ applications of $V_SV_A$, coherent states in intersecting modes will have amplitude $\sin((2l+1)\theta)\frac{\alpha}{\sqrt{k}}$, and coherent states in non-intersecting modes will have amplitude $\cos((2l+1)\theta)\frac{\alpha}{\sqrt{n-k}}$. This follows directly from \eqref{grovelution} and the coherent state mapping. Thus, after $r$ iterations of $V_S V_A$, coherent states in intersecting modes will have mean photon number $\sin^2((2r+1)\theta)\frac{\abs{\alpha}^2}{k}\geq \frac{1-k/n}{k}\abs{\alpha}^2$ and coherent states in non-intersecting modes will have mean photon number $\cos^2((2r+1)\theta)\frac{\abs{\alpha}^2}{n-k}\leq \frac{k}{n(n-k)}\abs{\alpha}^2$. Thus,
at least one of the $n-k$ non-intersecting modes $i$ will click with probability no greater than
\begin{align}\label{pcinv}
1-e^{-\abs{\alpha}^2\frac{k}{n(n-k)}(n-k)}=1-e^{-\abs{\alpha}^2\frac{k}{n}}.
\end{align}
No clicks occur with probability
\begin{align}\label{pctot}
e^{-\abs{\alpha}^2}.
\end{align}
Thus, when $x$ and $y$ intersect in $k$ locations, the probability that Alice sends Bob a non-intersecting index $i$ or that no clicks occur is upper bounded by $1-e^{-\abs{\alpha}^2\frac{k}{n}}+e^{-\abs{\alpha}^2}$, which justifies the above choice of repetition number $K$.

Now we describe Alice and Bob's procedure to implement $V_A$.
\pagebreak
\begin{samepage}
\begin{framed}
\textbf{Procedure to implement linear optics transformation $V_A$ corresponding to oracle call $U_A$}

First, Alice prepares $n$ auxilliary modes initialized to $\ket{0}$, so the state of her entire register is
\begin{align}
\bigotimes_{i=1}^n (\ket{\lambda_i \alpha}\ket{0}),
\end{align}
where $\bigotimes_{i=1}^n \ket{\lambda_i \alpha}$ is the resultant state from the previous step in the coherent state protocol. Then,
\begin{enumerate}
\item For each $i$ in which $x_i=1$, Alice swaps the $i$-th pair of modes $\ket{\lambda_i \alpha}\ket{0}\rightarrow \ket{0}\ket{\lambda_i \alpha}$ (and otherwise applies the identity map), and sends the entire state to Bob.
\item For each $i$ in which $y_i=1$, Bob flips the sign of the second mode corresponding to index $i$ using a phase shifter, and sends the entire state back to Alice.
\item Alice repeats the first step: For each $i$ in which $x_i=1$, she swaps the $i$-th pair of modes $ \ket{0}\ket{\lambda_i \alpha}   \rightarrow \ket{\lambda_i \alpha}\ket{0}$ (and otherwise applies the identity map). Alice then discards the $n$ auxilliary modes.
\end{enumerate}
\end{framed}
\end{samepage}
It is straightforward to show that this procedure implements $V_A$ exactly.


%
%
%
%
%
%
%
%
%

\subsection{Information leakage of coherent state version of distributed Grover search protocol}\label{qicohg}

In this section we bound the information leakage of our coherent state version of the distributed Grover search protocol using the following more general result: the information leakage of any protocol for which, conditional on fixed inputs $x, y$, the two parties exchange pure coherent states in a superposition of $n$ modes with constant total mean photon number $\abs{\alpha}^2$ over $r'$ rounds is $\mathcal{O}(r' \log n)$. We will show that this more general result implies the information leakage of our coherent state distributed Grover search protocol is $\mathcal{O}(\sqrt{n} \log n)$, just as in the original protocol.

For any pure state protocol, the information leakage for round $i$ is given by
\begin{align}
\QIC_i (\Pi, \mu) &= H(C_i | Y B_i) + H(C_i | X A_i)\nonumber\\
&\leq 2 H(C_i), \label{qic_entropy}
\end{align}
where the inequality follows from the data processing inequality. We now bound the quantity \eqref{qic_entropy} when the states exchanged are coherent states in a superposition of $n$ modes with constant total mean photon number $\abs{\alpha}^2$ by projecting the state $\rho_i^{C_i}$ exchanged on round $i$ onto telescoping neighborhoods of the total mean photon number $\abs{\alpha}^2$. Define the following partition of the nonnegative integers into disjoint sets:

\begin{align}
\Gamma_{0}&=\left\{k\in \mathds{Z}^+: \abs{k-\abs{\alpha}^2} \leq \Delta-1 \right\}\label{gamma0}\\
\Gamma_{j}&=\left\{k\in \mathds{Z}^+: j\Delta\leq \abs{k-\abs{\alpha}^2} \leq (j+1)\Delta-1\right\}\label{gamma1}\text{for every positive integer $j$}.\nonumber
\end{align}

For each $j$, let $\Pi_j$ be the projection onto the space of Fock states with total photon number lying in the set $\Gamma_j$. Then the set $\{\Pi_0,\Pi_1,\dots\}$ forms a measurement. Let $E_1$, $E_2$ be identical classical registers containing the measurement outcome. Define an isometry
{
\begin{align}
V=\sum_{j=0}^{\infty} \Pi_j \otimes \ket{j} \otimes \ket{j} \in \U(L^2(\mathds{R}))^{\otimes n},L^2(\mathds{R}))^{\otimes n} \otimes E_1 \otimes E_2)
\end{align}
}
Applying $V$ to $\rho_i^{C_i}$ yields

\begin{align}
H(C_i)_{\rho_i}&= H(C_i E_1 E_2)_{V\rho_i V^{\dagger}}\nonumber\\
&= H(E_1)+H(C_i | E_1)+H(E_2 | E_1 C_i)\nonumber\\
&\leq 2 H(E_1)+H(C_i | E_1)\nonumber\\
&\leq 2 H(E_1)+\sum_{j=0}^\infty \text{Pr}(E_1=j) \log\dim (\Pi_j),\label{cohdim1}
\end{align}

where the first equality follows from isometric invariance of entropy, the second from the chain rule, the first inequality from the data processing inequality and the fact that $H(E_2) = H(E_1)$, and the second inequality from the dimension bound along with the property that conditioning on a classical register is taking the average.

We now bound the quantity \eqref{cohdim1}. We first treat the term $H(E_1)$. 
It can be shown that $H(E_1)$ is no greater than the entropy of the $\abs{\alpha}^2$ Poisson distribution, which is finite and constant in $n$ (in fact, it is well-approximated by $\frac{1}{2} \log( 2\pi e \mumax)$ when $\mumax \gg 1$ \cite{doi:10.1137/1030059}).

Now we treat the second term of \eqref{cohdim1}. We make the choice $\Delta\geq (e^2-1)\abs{\alpha}^2$ because it simplifies the asymptotic analysis. In practice, one can optimize over $\Delta$. Under this choice, using Chernoff bounds,
\begin{align}
\Pr(E_1=j)&\leq e^{-\abs{\alpha}^2}\left(\frac{e\abs{\alpha}^2}{\abs{\alpha}^2+j\Delta}\right)^{\abs{\alpha}^2+j\Delta}\nonumber\\
&\leq e^{-j \Delta}\hspace{1in} \text{for all } j\geq0.
\end{align}
Using the same technique as was used to prove Theorem 1 of \cite{PhysRevA.89.062305} it can also be shown that
\begin{align*}
\log \dim(\Pi_0) &\leq (\abs{\alpha}^2+\Delta-1)\log (\abs{\alpha}^2+\Delta+n-2)+\log(2\Delta-1)\\
\log \dim(\Pi_j) & \leq (\abs{\alpha}^2+(j+1)\Delta-1)\log(\abs{\alpha}^2+(j+1)\Delta+n-2)+\log(2\Delta).
\end{align*}
Using these bounds it is straightforward to show that the second term of \eqref{cohdim1} is $\mathcal{O}(\log n)$. Thus, after $r$ rounds the total information cost is $\mathcal{O}(r\log n)$.

Now we apply this bound to the coherent state version of the Grover search protocol. This is a pure state protocol, and every state has total mean photon number $\abs{\alpha}^2$. This follows from $V_S \ket{\psi_\alpha} = f_\alpha(U_S \ket{\psi})$ and $V_A \ket{\psi_\alpha} = f_\alpha(U_A \ket{\psi})$ for every state $\ket{\psi}$ used in the original protocol, and that Alice and Bob's manipulations of the state to jointly perform $V_A$ do not change the total mean photon number. Each state communicated between Alice and Bob is a tensor product of $n$ coherent states. For $K$ repetitions, by straightforward application of Lemma~\ref{QICdiscard} ($\QIC$: increasing under discarding of side information), the fact that this protocol uses $Kr=\mathcal{O}(\sqrt{n/k})$ rounds of quantum communication, and the above information cost bound, the information cost of this stage is $\mathcal{O}(\sqrt{n/k}\log n)$.

For each repetition of the protocol, Alice sends Bob her measurement outcome $i$ (which is $\log n$ bits) along with $x_i$ (which is one bit), and Bob sends Alice $y_i$ (which is one bit). Or, if Alice received no clicks she uses one bit to tell Bob. Thus, the amount of communication in these stages is upper bounded by $K(\log n+2+1)=\mathcal{O}(\log n)$ bits, which also upper bounds the information leakage of these stages by the dimension bound. Thus, in total, this protocol has information leakage $\mathcal{O}(\sqrt{n/k}\log(n))$.

\subsection{Limiting the interaction}

We now wish to limit the interaction in the protocols. We limit how many modes must interact together to $r^2$, and we limit the number of rounds of interaction these modes undergo to $O(r)$.

We show how to adapt the protocol of the previous section, call it $\Pi_S$, achieving $O(\sqrt{n} \log n)$ leakage, to a protocol achieving $\frac{n}{r} \log r$ leakage, up to logarithmic terms, for $r \leq \sqrt{n}$. This protocol only requires interfering $r^2$ modes at once, and these modes are only exchanged for $r$ rounds.

\begin{framed}
\textbf{Protocol} $\hat{\Pi}_S$ on inputs $x, y \in \{0, 1 \}^n$:

\begin{itemize}
\item Divide inputs into $\frac{n}{r^2}$ blocks of size $r^2$
\item Run protocol $\Pi_S$ on each block
\item For each block outputting $i \in [n]$, exchange $x_i$ and $y_i$
\item If there exists at least one such pair $x_i$, $y_i$ such that $AND (x_i, y_i) = 1$, output smallest such $i$.
\item Else, output $\emptyset$.
\end{itemize}
\end{framed}

First, note that the error is at most the same as in $\Pi_S$ run on instances of size $r^2$. Second, by discarding quantum side-information (Lemma~\ref{QICdiscard}), information leakage is at most $\frac{n}{r^2}$ times that of $\Pi_S$ run on instances of size $r^2$, which is $\sqrt{r^2} \log r = r \log r$. Hence, the total leakage is $\frac{n}{r} \log r$. 

Note that we still need to interfere $r^2$ modes, and the effect of dark counts on $r^2$ modes combine, so that this is both more challenging experimentally and has worse error propagation than our protocol. Hence, we argue that our main protocol is more practical than this one.

\bibliographystyle{unsrt}
\bibliography{appt_sch}

\begin{thebibliography}{10}

\bibitem{raz1999exponential}
Ran Raz.
\newblock Exponential separation of quantum and classical communication
  complexity.
\newblock In {\em Proceedings of the thirty-first annual ACM symposium on
  Theory of computing}, pages 358--367. ACM, 1999.

\bibitem{gavinsky2007exponential}
Dmitry Gavinsky, Julia Kempe, Iordanis Kerenidis, Ran Raz, and Ronald De~Wolf.
\newblock Exponential separations for one-way quantum communication complexity,
  with applications to cryptography.
\newblock In {\em Proceedings of the thirty-ninth annual ACM symposium on
  Theory of computing}, pages 516--525. ACM, 2007.

\bibitem{Buhrman:1998:QVC:276698.276713}
Harry Buhrman, Richard Cleve, and Avi Wigderson.
\newblock Quantum vs. classical communication and computation.
\newblock In {\em Proceedings of the Thirtieth Annual ACM Symposium on Theory
  of Computing}, STOC '98, pages 63--68, New York, NY, USA, 1998. ACM.

\bibitem{HdW:2002}
Peter H{\o}yer and Ronald De~Wolf.
\newblock Improved quantum communication complexity bounds for disjointness and
  equality.
\newblock In {\em STACS}, pages 299--310. Springer, 2002.

\bibitem{aaronson:2003}
Scott Aaronson and Andris Ambainis.
\newblock Quantum search of spatial regions.
\newblock In {\em Foundations of Computer Science, 2003. Proceedings. 44th
  Annual IEEE Symposium on}, pages 200--209. IEEE, 2003.

\bibitem{klauck2001interaction}
Hartmut Klauck, Ashwin Nayak, Amnon Ta-Shma, and David Zuckerman.
\newblock Interaction in quantum communication and the complexity of set
  disjointness.
\newblock In {\em Proceedings of the thirty-third annual ACM symposium on
  Theory of computing}, pages 124--133. ACM, 2001.

\bibitem{JainRS03}
Rahul Jain, Jaikumar Radhakrishnan, and Pranab Sen.
\newblock A lower bound for the bounded round quantum communication complexity
  of {Set Disjointness}.
\newblock In {\em Proceedings of the 44th Annual IEEE Symposium on Foundations
  of Computer Science}, pages 220--229, 2003.

\bibitem{braverman2015near}
Mark Braverman, Ankit Garg, Young~Kun Ko, Jieming Mao, and Dave Touchette.
\newblock Near-optimal bounds on bounded-round quantum communication complexity
  of disjointness.
\newblock In {\em Proceedings of the 56th Annual IEEE Symposium on Foundations
  of Computer Science}, pages 773--791. IEEE, 2015.

\bibitem{PhysRevA.89.062305}
Juan~Miguel Arrazola and Norbert L\"utkenhaus.
\newblock Quantum fingerprinting with coherent states and a constant mean
  number of photons.
\newblock {\em Phys. Rev. A}, 89:062305, Jun 2014.

\bibitem{AT16}
Juan~Miguel Arrazola and Dave Touchette.
\newblock Quantum advantage on information leakage for equality.
\newblock {\em arXiv preprint arXiv:1607.07516}, 2016.

\bibitem{PhysRevA.95.032337}
Niraj Kumar, Eleni Diamanti, and Iordanis Kerenidis.
\newblock Efficient quantum communications with coherent state fingerprints
  over multiple channels.
\newblock {\em Phys. Rev. A}, 95:032337, Mar 2017.

\bibitem{PhysRevA.93.062311}
Juan~Miguel Arrazola, Markos Karasamanis, and Norbert L\"utkenhaus.
\newblock Practical quantum retrieval games.
\newblock {\em Phys. Rev. A}, 93:062311, Jun 2016.

\bibitem{PhysRevA.95.062334}
Ryan Amiri and Juan~Miguel Arrazola.
\newblock Quantum money with nearly optimal error tolerance.
\newblock {\em Phys. Rev. A}, 95:062334, Jun 2017.

\bibitem{grover1996fast}
Lov~K Grover.
\newblock A fast quantum mechanical algorithm for database search.
\newblock In {\em Proceedings of the twenty-eighth annual ACM symposium on
  Theory of computing}, pages 212--219. ACM, 1996.

\bibitem{boyer1996tight}
M.~{Boyer}, G.~{Brassard}, P.~{H{\o}yer}, and A.~{Tapp}.
\newblock {Tight Bounds on Quantum Searching}.
\newblock {\em Fortschritte der Physik}, 46:493--505, 1998.

\bibitem{bhattacharya2002implementation}
N~Bhattacharya, HB~van~Linden van~den Heuvell, and RJC Spreeuw.
\newblock Implementation of quantum search algorithm using classical fourier
  optics.
\newblock {\em Physical review letters}, 88(13):137901, 2002.

\bibitem{touchette2015quantum}
Dave Touchette.
\newblock Quantum information complexity.
\newblock In {\em Proceedings of the forty-seventh annual ACM symposium on
  Theory of computing}, pages 317--326. ACM, 2015.

\bibitem{LauriereT:2016}
Mathieu Lauri\`ere and Dave Touchette.
\newblock The flow of information in interactive quantum protocols :the cost of
  forgetting.
\newblock In {\em Proceedings of the 2017 Conference on Innovations in
  Theoretical Computer Science, To appear}, ITCS '17, 2017.

\bibitem{braverman2013information}
Mark Braverman, Ankit Garg, Denis Pankratov, and Omri Weinstein.
\newblock From information to exact communication.
\newblock In {\em Proceedings of the forty-fifth annual ACM symposium on Theory
  of computing}, pages 151--160. ACM, 2013.

\bibitem{dagan2016trading}
Yuval Dagan, Yuval Filmus, Hamed Hatami, and Yaqiao Li.
\newblock Trading information complexity for error.
\newblock {\em arXiv preprint arXiv:1611.06650}, 2016.

\bibitem{Barak:2010:CIC:1806689.1806701}
Boaz Barak, Mark Braverman, Xi~Chen, and Anup Rao.
\newblock How to compress interactive communication.
\newblock {\em SIAM Journal on Computing}, 42(3):1327--1363, 2013.

\bibitem{Braverman2011}
Mark Braverman and Anup Rao.
\newblock Information equals amortized communication.
\newblock In {\em Proceedings of the 52nd Symposium on Foundations of Computer
  Science}, FOCS '11, pages 748--757, Washington, DC, USA, 2011. IEEE Computer
  Society.

\bibitem{braverman2015interactive}
Mark Braverman.
\newblock Interactive information complexity.
\newblock {\em SIAM Journal on Computing}, 44(6):1698--1739, 2015.

\bibitem{cleve1999quantum}
Richard Cleve, Wim Van~Dam, Michael Nielsen, and Alain Tapp.
\newblock Quantum entanglement and the communication complexity of the inner
  product function.
\newblock In {\em Quantum Computing and Quantum Communications}, pages 61--74.
  Springer, 1999.

\bibitem{baumeler2015quantum}
{\"A}min Baumeler and Anne Broadbent.
\newblock Quantum private information retrieval has linear communication
  complexity.
\newblock {\em Journal of Cryptology}, 28(1):161--175, 2015.

\bibitem{chailloux2017information}
Andr{\'e} Chailloux, Iordanis Kerenidis, and Mathieu Lauri{\`e}re.
\newblock The information cost of quantum memoryless protocols.
\newblock {\em arXiv preprint arXiv:1703.01061}, 2017.

\bibitem{JainRS09}
Rahul Jain, Jaikumar Radhakrishnan, and Pranab Sen.
\newblock A new information-theoretic property about quantum states with an
  application to privacy in quantum communication.
\newblock {\em Journal of the ACM}, 56(6), September 2009.
\newblock Article no.~33.

\bibitem{klauck2002quantum}
Hartmut Klauck.
\newblock On quantum and approximate privacy.
\newblock {\em Lecture notes in computer science}, pages 335--346, 2002.

\bibitem{salvail2009power}
Louis Salvail, Christian Schaffner, and Miroslava Sotakova.
\newblock On the power of two-party quantum cryptography.
\newblock In {\em International Conference on the Theory and Application of
  Cryptology and Information Security}, pages 70--87. Springer, 2009.

\bibitem{kerenidis2016information}
Iordanis Kerenidis, Mathieu Lauriere, Fran{\c{c}}ois Le~Gall, and Mathys
  Rennela.
\newblock Information cost of quantum communication protocols.
\newblock {\em Quantum Information \& Computation}, 16(3\&4):181--196, 2016.

\bibitem{1611.06650}
Yuval Dagan, Yuval Filmus, Hamed Hatami, and Yaqiao Li.
\newblock Trading information complexity for error.
\newblock {\em arXiv:quant-ph/1611.06650}, 2016.

\bibitem{doi:10.1137/1030059}
Ronald~J. Evans and J.~Boersma.
\newblock The entropy of a poisson distribution (c. robert appledorn).
\newblock {\em SIAM Review}, 30(2):314--317, 1988.

\end{thebibliography}

\end{document}